\documentclass[12pt]{article}

\usepackage{amsmath}
\usepackage{amssymb}
\usepackage{amsthm}
\usepackage{authblk}
\usepackage{bm}
\usepackage{epstopdf}
\usepackage{epsfig}
\usepackage{rotating}
\usepackage{geometry}      
\usepackage{graphicx, caption}
\usepackage{pgfplots}
\usepackage{tikz}
\usepackage{marginnote} 
\usepackage{algorithm}
\usepackage{algpseudocode}
\usepackage{subfigure}
\usepackage{makecell}
\usepackage{lscape}
\usetikzlibrary{quotes,angles}
\usepackage{nameref} 
\usepackage{url}
\usepackage{natbib}

\newtheorem{prop}{Proposition}
\newtheorem{corr}{Corollary}

\title{Quantile Regression Modelling via Location and Scale Mixtures of Normal Distributions}
\author{Haim Y. Bar\thanks{
    haim.bar@uconn.edu. The authors gratefully acknowledge the following funding support: Prof. Bar's research was supported by NSF-DMS
1612625. Professor Booth's research was partially supported by NSF-DMS 1611893. Professor Wells' research was partially supported by NSF-DMS 1611893, and NIH grant U19 AI111143.}\\
    315 Philip E. Austin Building
Department of Statistics,
University of Connecticut
Storrs, CT, 06269-4120, USA.\\
    and \\
    James G. Booth and Martin T. Wells\\
    Department of Statistics and Data Science, Cornell University, \\ Ithaca NY, 14853, USA.}

\pgfplotsset{compat=1.12}

\date{}	
\begin{document}
\maketitle
\vspace{-1cm}
\begin{abstract}
We show that the estimating equations for quantile regression can be 
solved using a simple EM algorithm in which the M-step is computed via
weighted least squares, with weights computed at the E-step as the
expectation of independent generalized inverse-Gaussian variables. We compute the
variance-covariance matrix for the quantile regression coefficients
using a kernel density estimator that results in more
stable standard errors than those produced by existing software. A
natural modification of the EM algorithm that involves fitting a linear
mixed model at the M-step extends the methodology to mixed effects
quantile regression models. In this case, the fitting method can be
justified as a generalized alternating minimization algorithm. 
Obtaining quantile regression estimates via the weighted least squares
method enables model diagnostic techniques similar to the ones used 
in the linear regression setting.
The computational approach is compared with existing software using
simulated data, and the methodology is illustrated with several case
studies. 
\end{abstract}



\section{Introduction}\label{sec.introduction}
Quantile regression (QR) is used to predict how a certain percentile of a quantitative response variable, $Y$,
changes with some predictors, $\mathbf{x}=(x_1,\ldots,x_p)$ and as a consequence offers an approach for examining how covariates influence the
location, scale, and shape of the response distribution. It is assumed that the predictors $x_j$ 
and the $q$-th quantile of $Y$ are linearly related, and that $y_1,\ldots,y_n$ are 
independent distributed
as $P(y_i<y)=F(y-\mathbf{x}_i^T\bm{\beta}_q)$, where $\bm{\beta}_q$ is a vector of unknown
coefficients which depends on $q\in(0,1)$. 
(Henceforth, as usual, $\mathbf{x}_i^T\bm\beta_q$ will include the intercept, $\beta_{0,q}$.)

The idea of estimating a median regression slope based on minimizing sum of the absolute deviances has a long history that dates back to Boskovic, Edgeworth, and Laplace.  A very special case of QR is when $q=0.5$ and the random errors $y_i - \mathbf{x}_i^T\bm\beta$
are assumed to be (i) normally distributed, (ii) with mean 0 and variance $\sigma^2$, and
 (iii) uncorrelated with the predictors.  In this case, the relationship between $\mathbf{x}$ and the 
 median of $Y$
 is the same as between $\mathbf{x}$ and the mean of $Y$, and linear regression (LR) is used to estimate the
 rate of change in $Y$ for a unit change in each $x_j$, via the ordinary least squares (OLS) method.
Yet, although QR is more versatile in the sense that it can be used to make more general
 predictions than just for the mean of $Y$, and although LR relies on more restrictive 
 and sometimes unrealistic assumptions, and is sensitive to outliers, there is no doubt that for
 over a century LR has been vastly more popular in applications than QR.
This is, perhaps, due to the fact that LR offers closed-form formulas for the model parameters,
and also because, when the assumptions regarding the random error are valid, the distributional properties
of the model parameters are known and can be used to perform inferential procedures. Inference in
the QR setting is based on rank-based methods which are less powerful than their parametric 
counterparts in LR, or on computationally challenging resampling methods, such as the bootstrap.
Another historical reason for favoring the linear regression framework is that it 
extends rather naturally to include complex error structures such as random effects.

In the LR framework, the normal model for the errors yields a convenient likelihood function which
can be easily maximized with respect to $\bm\beta$. In practice, this is achieved through solving the estimating
equation which is obtained by minimizing the negative of the log-likelihood, that is, by solving
$$\hat{\bm\beta}=\underset{\bm{\beta}}{\arg\min}\sum_{i=1}^n(y_i-\mathbf{x}_i^T\bm\beta)^2\,.$$
So, numerically, the solution is obtained by minimizing the quadratic loss function, rather than
directly maximizing the likelihood.
Similarly, in the QR framework, the estimation of the regression parameters (for a specific quantile of interest, $q$)
is done (e.g., \cite{koen:hall:2001,koen:bass:1978}) by solving an estimating equation involving 
the `check' loss function (see Figure \ref{checkloss}):
\begin{equation}\label{objectiveFunction}
 \hat{\bm{\beta}}_q=\underset{\bm{\beta}}{\arg\min}\sum_{i=1}^n\rho_q\left(y_i-\mathbf{x}_i^T\bm{\beta}\right)
\end{equation}
where
\begin{equation}
 \rho_q(u)=u\cdot(q-1_{[u<0]})\,,
\end{equation}
and $1_{[u<0]}$ is the indicator function which equals 1 when the argument, $u$, is negative.
Equation (\ref{objectiveFunction}) does not lead to a closed-form formula for $\bm\beta_q$,
but a numerical solution via algorithmic methods is feasible. For a comprehensive review of QR in general,
see \cite{koen:2005}. For a different approach for estimating $\bm\beta_q$ using
the majorization-minimization (MM) algorithm, see \cite{HunterLange:2000}.

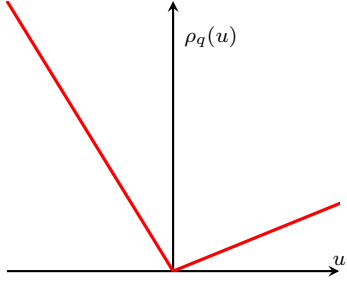
\begin{figure}
    \begin{center}
    \begin{tikzpicture}[scale=1 ]
    \pgfmathsetmacro{\a}{0.1}
    \pgfmathsetmacro{\b}{1} 
    \pgfplotsset{ticks=none}
      \begin{axis}[axis lines=middle,xmin=-0.75,xmax=0.75,ymin=0,ymax=0.75*0.8, thick, width=6cm,
        xlabel=$\scriptstyle u$,
        ylabel=$\scriptstyle \rho_q(u)$,
        x label style={at={(1,0.1)},anchor=north},
        y label style={at={(0.5,0.95)}},
       no marks,
       samples=100
            ]
        \coordinate (O) at (0,0);
    \addplot+[red,domain=-0.75:0,very thick] {-0.8*x};
    \addplot+[red,domain=0:3,very thick] {0.2*x};
          \end{axis}
    \end{tikzpicture}
    \end{center}      
    \caption{The `check' loss function ($q=0.2$)\label{checkloss}}
    \end{figure}

The optimization problem (\ref{objectiveFunction}) is equivalent to maximum likelihood estimation
under the assumption that the errors in the linear model,
$y_i = \mathbf{x}_i^T\bm\beta_q + u_i$,
are an iid sample from an asymmetric Laplace distribution (ALD), with
density given by
\begin{eqnarray*}
  f(u|\alpha,q)=q(1-q)\exp\left\{ -\rho_q( u/\alpha ) \right\}/\alpha
\end{eqnarray*}
\cite{yu:moye:2001}. 
Methods for fitting quantile regression models for linear mixed models using the asymmetric Laplace distribution
have been proposed; see, for example, \cite{gera:bott:2007},
\cite{gera:bott:2014}, and \cite{gala:band:lach:2017} who use an EM-type algorithm combined with numerical quadrature for fitting the QR model.
Although it is possible to represent QR in terms of a likelihood function,
the estimation procedure still requires numerical estimation of $\bm\beta_q$ and
the inferential procedures are computationally inefficient because
the loss function is not differentiable at 0. In a Bayesian approach to quantile regression, \cite{kozumi2011gibbs} posited the ALD as the working likelihood to perform the inference and \cite{yu:moye:2001} proposed an efficient Gibbs sampling algorithm by using a location and scale mixture representation of the ALD.


We propose a new computational method for QR, using an implicit
location and scale mixture model approach as
in \cite{kozubowski2000asymmetric}.  The computations needed to
minimize (\ref{objectiveFunction}) are typically formulated as a
linear programming problem
(\cite{koenker1987algorithm,portnoy1997gaussian}) which allows for
efficient computation. The utility of the location and scale mixture model representation used in this paper is that,
conditional on a particular set of random variables,  flexible  approaches  of a normal  distribution modeling framework are at our disposal.  The representation sets the stage for developing more complex QR models (see Section \ref{sec.casestudy}).

The underlying estimating equation we use is also (\ref{objectiveFunction}), so the QR estimator obtained from our method is identical to the one obtained via traditional estimation, and therefore, it is consistent. However, our approach has a few advantages over the traditional method. First, it allows us to obtain a closed-form analytic
solution to the regression problem. In fact, the solution is the well-known weighted least squares (WLS) estimator. Second, 
since our QR estimator is obtained via the weighted least squares method, it is possible to construct model
diagnostic techniques akin to the ones used in the LR setting.
Third, our method yields  less variable estimates for the standard errors of the regression
coefficients. Fourth, because the residuals are shown to be (conditionally) normally distributed, our method is
easily extended to quantile regression in mixed models. 
Finally, parameter estimation is done efficiently via an Expectation Maximization (EM) algorithm.  A key distinction between our approach to using the ALD and the Bayesian approach of \cite{yu:moye:2001} is that we do not posit the ALD as the working likelihood.  Our analytic representation is used as an exact representation of the QR estimating equation.  As a consequence the properties of our estimators will not depend on the correct data generating process. 

We note that the scale mixture of normal distributions approach was used
by \cite{AlbertChib:1993} in the context of linear models with binary or polychotomous response,
and by \cite{pols:scot:2012} in the context of estimating the coefficients of a separating hyperplane
in classification problems using support vector machines. The latter is perhaps most closely related to our approach,
in that they use a similar loss function, namely $d(\bm\beta)=\sum_{i=1}^n\max(1-y_i\mathbf{x}_i^T\bm\beta,0)$,
which is also non-differentiable at 0.

The paper proceeds as follows. In Section \ref{sec.model} we introduce our model and in Section \ref{sec.properties}
we derive some of its useful and interesting properties. In Section \ref{sec.mixedmodels} we discuss how the method is
extended to include random effects. In Section \ref{sec.simulation} we summarize our simulation study, and in Section
\ref{sec.casestudy} we apply our method to three data sets. In the first example we use a fixed-effect model to 
determine whether the level of some metabolites are associated with different BMI percentiles.
In the second example, we show how we can use our method to fit a longitudinal QR model, using 
temperature data. In the third case study, we show how to fit a frailty QR model, using emergency department
length of visit data.
Section \ref{sec.discussion} contains a brief conclusion and discussion of future work.

\section{An Implicit Location and Scale Mixture Model for Quantile Regression}\label{sec.model}
We introduce a scale-mixture of normal distributions model
for $\exp[-2\rho_q(u_i)]$ and show that it yields
quantile regression estimates which have the familiar weighted least squares form, and can be obtained
efficiently via an EM algorithm \cite{EM}.
\begin{prop}\label{propPLik}
Define the joint distribution of $(u, \lambda)$ as follows.
Suppose that $\lambda$ has an 
exponential distribution with rate equal to $2q(1-q)$ and that 
$u|\lambda\sim N[(1-2q)\lambda,\lambda]$. Then the marginal density of
$u$ is $h(u) = 2q(1-q)e^{-2\rho_q(u)}$.
\end{prop}

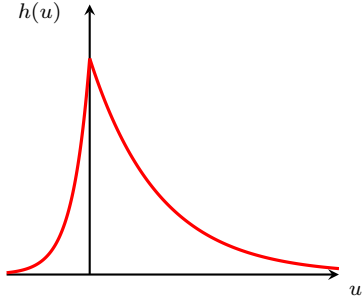
\begin{figure}
    \begin{center}
    \begin{tikzpicture}[scale=1 ]
    \pgfmathsetmacro{\a}{0.1}
    \pgfmathsetmacro{\b}{1} 
    \pgfplotsset{ticks=none}
      \begin{axis}[axis lines=middle,xmin=-3,xmax=9,ymin=0,ymax=0.4, thick, width=6cm,
        xlabel=$\scriptstyle u$,
        ylabel=$\scriptstyle h(u)$,
        x label style={at={(1.05,0)},anchor=north},
        y label style={at={(0,1.05)}},
       no marks,
       samples=100
            ]
        \coordinate (O) at (0,0);
    \addplot+[red,domain=-3:0,very thick] {2*0.2*0.8*exp(2*0.8*x)};
    \addplot+[red,domain=0:9,very thick] {2*0.2*0.8*exp(-2*0.2*x)};
          \end{axis}
    \end{tikzpicture}
    \end{center}      
    \caption{The marginal density, $h(u)=\int_0^\infty f(u|\lambda)\pi(\lambda)d\lambda$, where $u|\lambda\sim N[(1-2q)\lambda,\lambda]$ and  $\lambda\sim Exp(2q(1-q))$.
    (In this case, showing $h(u)$ for  $q=0.2$)\label{hu}}
    \end{figure}
    
The proof appears  in the appendix. The shape of the marginal density function $h(u)$ is depicted in Figure \ref{hu}.
Setting the residuals from the quantile regression
to be $u_i=y_i-\mathbf{x}_i^T\bm\beta_q$,
and allowing the distribution of each $u_i$ to have its own scale parameter, $\lambda_i$,
Proposition (\ref{propPLik})
implies that the minimizer, $\hat{\bm\beta}_q$ in (\ref{objectiveFunction}) can be obtained 
using an EM algorithm in which the exponentially distributed $\lambda_i$'s are missing values.
Specifically, the joint density of $(u_i,\lambda_i)$ 
is
$$p(u_i,\lambda_i | \bm\beta_q)= \frac{2q(1-q)}{\sqrt{2\pi\lambda_i}}\exp\left\{ -\frac{[u_i-(1-2q)\lambda_i]^2 }{ 2\lambda_i } \right\}\exp\{ -2q(1-q)\lambda_i \}\,.$$
Hence the complete data log-likelihood (ignoring terms not involving
$u_i$, and hence also $\bm\beta$) is given by
$$ \sum_{i=1}^n\left( -\frac{u_i^2}{2\lambda_i} + (1-2q)u_i \right) \,. $$
The \textbf{E-step} of the EM algorithm involves taking expectation with
respect to the conditional distribution of $\lambda^{-1}_i$ given
$u_i$, which is inverse Gaussian, $\mathcal{IG}(|u_i|^{-1},1)$. 
Thus, in the \textbf{M-step} we fix the values of $\hat\lambda_i^{-1}$ at their estimated expected values after the latest iteration,
and use the (conditional) multivariate normal distribution of $\mathbf{u}|\bm\lambda$ to obtain the weighted least squares (WLS) estimator
\begin{equation}\label{betahat}
\hat{\bm\beta}_q=(X^T\hat\Lambda^{-1}X)^{-1}X^T\hat\Lambda^{-1}\left[Y-(1-2q)\hat{\bm\lambda}\right]\,.
\end{equation}
The details are summarized in Algorithm \ref{algorithm}, where $\ell$ is the conditional
log-likelihood of $\mathbf{u}|\bm\lambda$, and to initialize $\bm\beta'_q$ we use the ordinary least square
estimator.

\begin{algorithm}
\caption{The Quantile Regression EM (QREM) algorithm}\label{algorithm}
\begin{algorithmic}[0]
\State Initialize $\epsilon>0$, $\delta = 2\epsilon$
\State Initialize $\hat{\bm\beta}^{'}_q=(X^TX)^{-1}X^TY$
\While{$\delta > \epsilon$}
\State E-Step: $\hat\lambda_i^{-1}=|y_i-\mathbf{x}_i^T\hat{\bm\beta}_q^{'}|^{-1}$
\State M-Step: $\hat{\bm\beta}_q=(X^T\Lambda^{-1}X)^{-1}X^T\Lambda^{-1}\left[Y-(1-2q)\hat{\bm\lambda}\right]$
\State $\delta=|\ell(\mathbf{u}|X, Y,\bm\beta^{'}_q)-\ell(\mathbf{u}|X, Y,\bm\beta_q)|$
\State $\bm\beta^{'}_q\leftarrow\bm\beta_q$
\EndWhile
\end{algorithmic}
\end{algorithm}

It is interesting to note that our representation of the asymmetric Laplace distribution falls within the large class of generalized hyperbolic distributions, which includes the Student's $t$, Laplace, hyperbolic, normal-inverse Gaussian, and variance-gamma distributions \cite{mcneil2005}.  The  generalized inverse Gaussian (GIG) distribution was introduced by \cite{good53}.  A generalized hyperbolic random variable $U$ can be represented by combining a random variable $\Lambda \sim \text{GIG}(\psi,\chi,\lambda)$ and a latent multivariate Gaussian random variable $ X \sim N(0,\Sigma)$ using the relationship $ U = \mu+ \Lambda \alpha + \sqrt{\Lambda} X,$
and it follows that $ U \mid \Lambda \sim N(\mu+\Lambda \alpha,\Lambda \Sigma)$.

\section{Properties of the Model}\label{sec.properties}
We now show that the QR estimator derived from our mixture model has some useful properties
for statistical inference and model diagnostics. First, we claim the following:
\begin{prop}\label{propConsistent}
The QR estimator, $\hat{\bm\beta}_q$, obtained from our model is consistent.
\end{prop}
\begin{proof} 
This is clear because,
minimizing  the check-loss function $\rho_q(u)$ yields the same solution obtained  by maximizing $\exp[-2\rho_q(u)]$.
See \cite{koen:2005}, Section 4.1.2, for a thorough discussion about the consistency of the QR estimator from
(\ref{objectiveFunction}). For convenience, 
the regularity conditions needed to ensure consistency can be found in Appendix \ref{sec.appendix}.
\end{proof} 

Next, we discuss the estimation of the standard errors of the regression coefficients.
We derived the estimate for $\bm\beta_q$ from the conditional multivariate normal distribution of 
$\mathbf{u}|\bm\lambda$  but, 
the naive covariance of the WLS estimator, $(X^T\Lambda^{-1}X)^{-1}$ with $\lambda_i$ 
 obtained via the EM-algorithm, is not valid for inference, for two reasons. First,
the error distribution is unspecified so there is no `data generating model'
from which the covariance of $\hat{\bm\beta}_q$ can be estimated.
Second, the mean of $y_i$, $\mathbf{x}_i^T\bm\beta_q+(1-2q)\lambda_i$ is a function of the variance
$\lambda_i$.
Therefore, to obtain the large-sample covariance matrix of the QR parameters, we use Bahadur's
representation \cite{bahadur1966} to establish the following:
\begin{prop}
Let $u_i$ be the QR residuals and assume that they are independent and identically distributed, with
$q$-quantile equal to $0$. Then,
$$a.var(\hat{\bm\beta}_q)=\frac{q(1-q)}{f(0)^2}(X^TX)^{-1}$$
where $f(0)$ is the probability density of the $u_i$'s at 0.
\end{prop}

To obtain $f(0)$ we use a kernel density estimator, which is straightforward and fast, 
and does not require resampling methods. The same form of the covariance estimator
is used by \cite{Kocherginsky:2005}, but they propose using a Markov chain marginal bootstrap 
approach.  The performance of our approach is studied in the simulations section.
For more about Bahadur's representation for quantile regression, see also \cite[Chapter~4.3]{koen:2005}.

We now turn characteristics of the QR residuals, and derive properties which can be used to perform model
diagnostics, very similar to those used in linear regression for the mean model.
In Proposition \ref{propPLik} we derived the marginal distribution of the residuals $u_i$. From that,
the following properties of the residuals are easily obtained:
\begin{corr}\label{corrALD}
Marginally, $u_i\sim ALD(0, 2\sqrt{q(1-q)}, \sqrt{q/(1-q)})$ where $ALD$ is the asymmetric Laplace
distribution with three parameters (location, scale, and asymmetry), thus,
\begin{enumerate}
\item its cumulant generating function is
$K_U(t) = log[4q(1-q)] - log\left( [2(1-q)+t][2q-t] \right)$
\item $E(u)=\frac{1-2q}{2q(1-q)}$
\item $Var(u)=\frac{1}{4}\left(\frac{1}{q^2} + \frac{1}{(1-q)^2}\right)$
\item If $x_p$ is the $p$th quantile of the QR residuals, then
\begin{eqnarray*}
  x_p &=& \frac{1}{2(1-q)}\ln\frac{p}{q}\hspace{3mm}\mbox{if } p<q \\
  x_p &=& \frac{1}{2q}\ln\frac{1-q}{1-p}\hspace{3mm}\mbox{if } p>q \\
\end{eqnarray*}
\item $\int_{-\infty}^0 h(u)du = q$ and $\int_0^\infty h(u)du = 1-q$.
\end{enumerate}
\end{corr}

Furthermore, the fact that our QR estimator has the WLS closed-form as in (\ref{betahat}) yields
the following:

\begin{prop}\label{propUncorr}
Let $\mathbf{c}=sgn(Y-X^T\bm\beta_q) - (1-2q)\mathbf{1}$ be the scaled, binary-valued QR
residuals. Then, $X^T\mathbf{c}=\mathbf{0}$.
\end{prop}

The proof appears  in the appendix. This is similar to linear regression with normal errors, where the errors, 
$\mathbf{e}$, are uncorrelated with the explanatory
 variables, namely $E(\mathbf{x}_j\cdot \mathbf{e})=0$.
To use this result to produce model diagnostic plots, consider first a continuous predictor, $\mathbf{x}_j$.
Let $A=\left\{i : c_i = 2q \right\}$ and $B=\left\{i : c_i = 2q-2 \right\}$ be the sets of points above and below
the quantile regression line, respectively.
Under the correct regression model, Proposition \ref{propUncorr} implies that for each predictor
the density of points above the QR line is the same as for the points below.
So, we construct a QQ-plot for the $j$-th predictor by plotting the quantiles of $x_{j,i}$ for $i\in A$ versus the quantiles of $x_{j,i}$ for $i\in B$. If the model is adequate, the points should lie close to the
$45^{\circ}$ line. We demonstrate such QQ-plots in the sections \ref{sec.simulation} and \ref{sec.casestudy}.

For categorical predictors, Proposition \ref{propUncorr}  implies that under the correct model, for each level, $k$, 
we expect to have $P(c_i=2q | x_{j,i}=k)=1-q$. 

Analogously to the linear model framework, where the regression parameters are obtained by minimizing the
quadratic loss function and the goodness of fit of competing models is measured in terms of the mean
sum of squared errors, we obtain the regression parameters by minimizing the check loss function, and hence,
a natural measure of goodness of fit is the mean  check loss error:
$$\bar{G}_{\rho_q}=\frac{1}{n}\sum_{i=1}^n\rho_q(u_i)\,.$$
Or, we can define $G=2n\bar{G}_{\rho_q}$, and, since $G=2\sum_{i=1}^n\rho_q(u_i)=-\sum_{i=1}^n \log(h(u_i))$,
the goodness of fit criterion can be defined in terms of the marginal distribution of $u$,
$$G=-\log(\mbox{maximum marginal likelihood of $u$})\,.$$

\section{Extension to Mixed Models}\label{sec.mixedmodels}
Many applications involve dependence between observations, which is not captured by the
fixed-effects model. A convenient modeling approach in such cases is to assume that dependent observations
come from a multivariate normal distribution. Often, a certain correlation structure is
assumed, if it can be argued that it  corresponds to the setting of the experiment.
With this approach, the relationship between the response and the predictors is
expressed as a mixed model
\begin{equation}
\mathbf{y} = X^T \bm\beta +Z^T\mathbf{v} + \bm\epsilon\,,
\end{equation}
where $\mathbf{v}$ are the random effect parameters, so that $\mathbf{v}\sim N(0, K)$.
It is assumed that the random errors are independent from the random effects and
$\bm\epsilon\sim N(0,\sigma_e^2I)$.

\cite{koen:2005} refer to longitudinal studies and says that
`...\textit{many of the tricks of the trade developed for Gaussian random effects models are
no longer directly applicable in the context of quantile regression}.'
This is due to the fact that the regression estimates are not obtained via the
convenient and tractable least squares estimation method.
However, with our approach it is possible to extend statistical methods which apply
to mean regression in mixed models to QR. We drop the assumption that the random errors
have to be normally distributed and only require that they are conditionally independent, and are also
independent from the random effects. 

The extension of our method to mixed models is, in  principle, straightforward.
Let the residuals be $\mathbf{u}=\mathbf{y}-X^T\bm\beta_q - Z^T\mathbf{v}$.
Similarly to the fixed effect model, we assume the following hierarchical model
\begin{eqnarray*}
 u|\lambda, \mathbf{v} &\sim& N((1-2q)\lambda, \lambda)\\
 \lambda &\sim& Exp(2q(1-q))
\end{eqnarray*}
and as a result, 
$$\lambda^{-1}|y, \mathbf{v} \sim IG(|u|^{-1}, 1)\,.$$

The key difference relative to the fixed effect model is that in order for $\lambda^{-1}$ to 
follow an inverse Gaussian distribution, we have to condition not only on $y$, but also on
$\mathbf{v}$.
To justify plugging in the best linear unbiased predictor (BLUP) for $\mathbf{v}$, we rely on \cite{Gunawardana:2005} and their 
generalization of the EM algorithm, called GAM (Generalized Alternating Minimization).
Like the EM algorithm, GAM consists of two steps. The `backward' step generalizes the M-step, and the
`Forward' step generalizes the E-step.
In the case of our scale-mixture of normal distributions model, $\bm\beta$ and $K$ are the parameters
and we treat $\bm\lambda$ and $\mathbf{v}$ as the missing data.
Because the complete data likelihood can be written in closed-form, the backward step in our case
is identical to the M-step, and $\hat{\bm\beta}$ 
and $\hat{K}$ are the maximum likelihood estimators given the current imputed values of $\bm\lambda$ and $\mathbf{v}$.

Regarding the Forward step,  \cite{Gunawardana:2005} refer to a probability 
distribution function, $Q_C$, as `desired' if it has the properties that (i) the maximum likelihood is obtained
with the observed data, and (ii) it reduces the Kullback-Leibler divergence relative to the previous iteration.
They denote the set of desired distribution functions by $\mathcal{D}$.
The objective in the forward step if to find $Q_C\in\mathcal{D}$ such that
\begin{equation}\label{KL}
 D_{KL}(Q_C^{(t+1)}||P_{C}(\hat{\bm\beta}^{(t)},\hat{K}^{(t)})) \le D_{KL}(Q_C^{(t)}||P_{C}(\hat{\bm\beta}^{(t)},\hat{K}^{(t)}))\,,
\end{equation}
where $P_{C}(\hat{\bm\beta}^{(t)},\hat{K}^{(t)})$ is the member of the parametric family of the complete data
likelihood, evaluated at the MLE's after iteration $t$.
Note that the objective in the EM algorithm is to find a $Q_C$ which minimizes the KL divergence 
on the left-hand side of
(\ref{KL}), while in order to guarantee the convergence of a GAM procedure it is sufficient to find any
desired distribution which reduces it.

While finding a simultaneous update for $\mathbf{v}$ and $\bm\lambda$ may be intractable, 
the Forward step can be performed in two steps, and the convergence of GAM will still hold.
Let $\mathcal{D}_{|\mathbf{v}}\subset\mathcal{D}$ be the set of desired distributions which satisfy (\ref{KL})
while holding $\bm\lambda$ fixed at their current value. Clearly, the BLUP not only reduces the KL
divergence, but it actually minimizes it, making the BLUP
 the optimal next estimate for $\mathbf{v}$, given $\bm\lambda$.
Now, let $\mathcal{D}_{|\bm\lambda}\subset\mathcal{D}$ be the set of desired distributions which satisfy (\ref{KL})
while holding $\mathbf{v}$ fixed at their current value (the BLUP). Then, given $\mathbf{v}$ and
$y_i$, the best update for $\lambda_i^{-1}$
is $|u_i|^{-1}$, as we have shown previously in the fixed-effect model.

This two-step approach is
valid because (i) any distribution function obtained from a projection of $\mathcal{D}$ to a subspace is also
a valid candidate for $Q_C^{(t+1)}$ when using $\mathcal{D}$, because GAM does not require finding the minimizer 
-- just an improvement
with respect to the previous iteration; and (ii) with each of the two projections into $\mathcal{D}_{|\mathbf{v}}$ and 
$\mathcal{D}_{|\bm\lambda}$ the conditions of the GAM convergence theorem in \cite{Gunawardana:2005}, hold.

The implementation of our method is simple. Because the scale-mixture of normals led to a
conditional normal distribution of $(\bm\beta^T_q, \mathbf{v}^T)$
 the estimation of the QR model parameters in the M-step
can be done by using existing tools which were built for the mean-model setting with normal errors,
such as \texttt{lm()} \cite{R:2018} for fixed effect models, or  \texttt{lmer()}
\cite{lme4:2015} for mixed models, or equivalent procedures in other languages such as SAS and Stata.
This also facilitates the estimation of the matrix $K$  when fitting mixed models.

\section{Simulations}\label{sec.simulation}
Table \ref{sims} in Appendix \ref{sec.appendix} contains details regarding
25 different simulation scenarios, performed in order to assess the performance of QREM in terms of
bias and variance of regression parameter estimates.
In some simulations the error variance did not satisfy the usual mean-model regression assumptions,
namely, being normally distributed and uncorrelated with the predictors.
In scenarios 14-18 the error variances depend on a predictor, and in 19-22 the error terms were sampled
from a skewed distribution (log-normal). 

We compared the estimated regression coefficients with those obtained from the \texttt{quantreg} package
\cite{quantreg:2018},
which uses a different estimation approach (namely, direct
minimization of the loss function
in \ref{objectiveFunction}.)
Since our model is derived from the same loss function it is expected that the two methods would
give similar estimates, and this is confirmed by our simulations. Indeed, the parameter
estimates from both methods are nearly identical (the small differences are attributed to the chosen
tolerance level of the computational methods and the fact that
empirical quantiles are not uniquely defined).
The variances of the regression parameter estimates from the two methods, however, are quite different.
Recall that we obtain the asymptotic covariance of $\hat{\bm\beta}$ by using Bahadur's representation,
which requires the estimation the density of the residual $u_i$ at 0. To do that, we use
kernel density estimation, as implemented in the \texttt{KernSmooth} package \cite{KernSmooth}.
See \cite{Deng2014DensityEI} for a review of kernel density estimation packages.
The  \texttt{quantreg} package computes confidence intervals using the inversion of a rank test, per
\cite{koenker1994confidence}.
Figure \ref{sdbeta1} shows $\hat\sigma_{\hat\beta_1}$ for $q\in\{0.05,0.1,\ldots,0.95\}$,
using QREM (blue squares, left), the bootstrap (dark-red triangle, middle), and \texttt{quantreg}
(orange circles, right). The horizontal yellow lines above each quantile represent the true estimate
of the standard deviation of $\hat\beta_1$, 
as obtained from  the asymptotic Bahadur-type estimation using the true density of $u_i$ at 0.
Our estimator has a smaller sampling variance across all scenarios and all quantiles.
Figures \ref{sdbeta1sim15} and \ref{sdbeta1sim21}  in Appendix \ref{sec.appendix} 
show the estimated standard deviations of $\hat\beta_1$ obtained from QREM
and \texttt{quantreg} for scenarios 15 and 21, respectively.
Again, the sampling variance of $\hat\sigma_{\hat\beta_1}$ obtained from QREM is smaller than
the one from \texttt{quantreg}, and especially near the edges ($q<0.15$ and $q>0.85$.)
These results show that while on average the coverage probability of the two methods is
expected to be similar, the smaller variance obtained from the QREM procedure imply that results
from this method are more stable and provide more reliable inference. The wider range
of variance estimates obtained from \texttt{quantreg} imply that this method is much more likely
to lack power or to be over-powered. 

\begin{figure}
  \begin{center}
 \includegraphics[ scale=0.8]{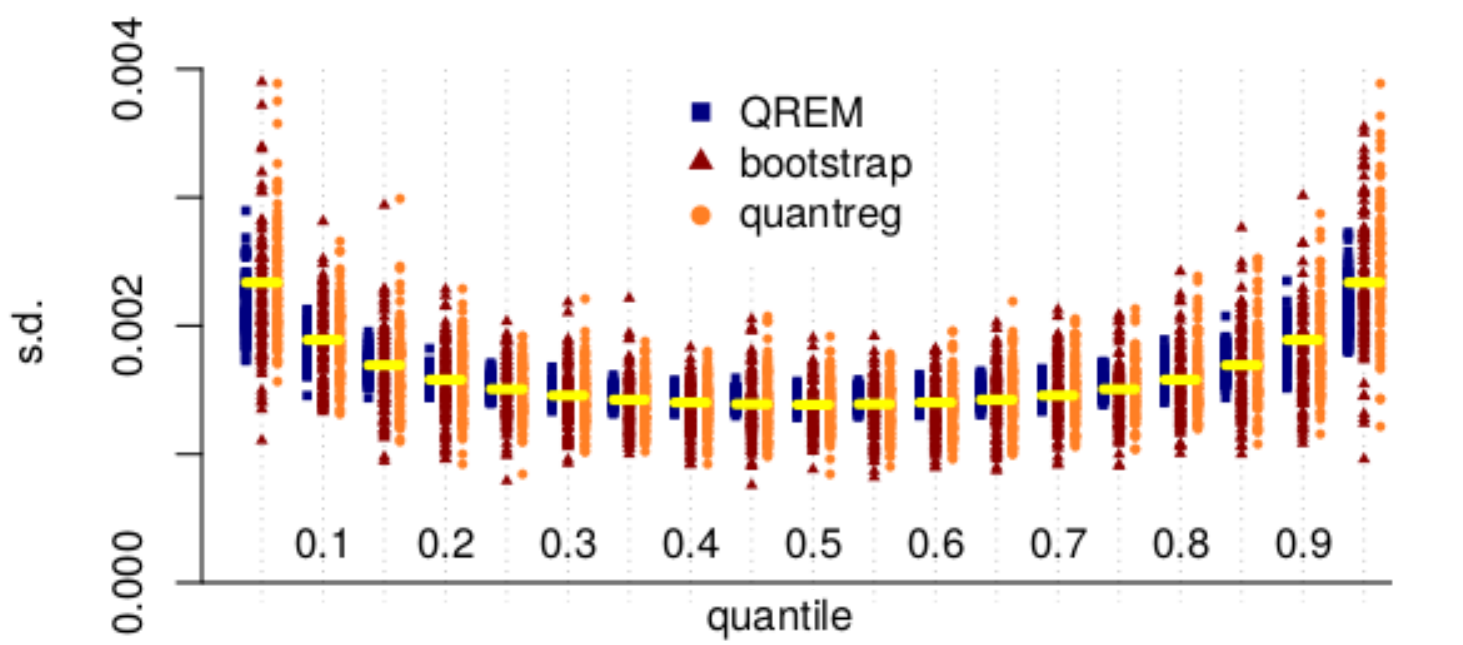}
\end{center}
\caption{QR simulation: $\hat\sigma_{\hat\beta_1}$ for $q\in\{0.05,0.1,\ldots,0.95\}$ for
$y\sim N(-3+x, 0.1^2)$.}\label{sdbeta1}
 \end{figure}

\begin{figure}[t!]
  \begin{center}
 \includegraphics[trim={1cm 0.5cm 1cm 2cm},clip,scale=1]{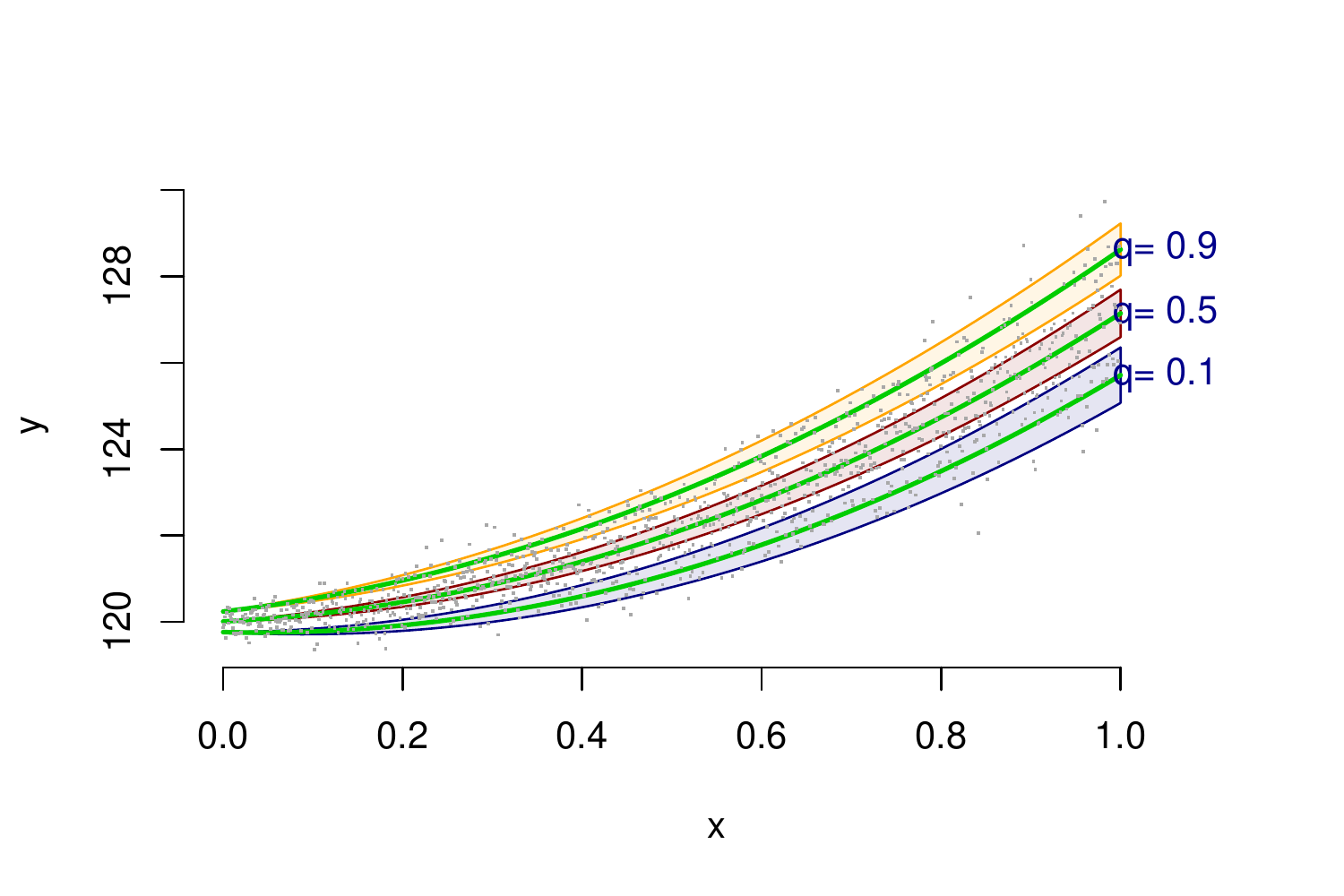}
\end{center}
\caption{QR simulation - the 0.1, 0.5, and 0.9 quantiles
with 90\% confidence intervals, when the true model is $y \sim N(6x^2+x+120, (0.2+x)^2)$.}\label{quadraticQ}
 \end{figure}

 We also simulated data in which the response depends on predictors in non-linear ways (simulations 23 and 24).
 For example, for simulation 23, in
Figure \ref{quadraticQ} we show results when $y \sim N(6x^2+x+120, (0.2+x)^2)$ so that the relation between
the mean of $y$ and $x$ is a quadratic function, and the standard deviation increases linearly
with $x$.
The figure shows the 0.1, 0.5, and 0.9 quantiles with 90\% confidence bands.
For all values of $x$, the three quantile regression curves are significantly different
from the each other. 
To assess goodness of fit, we use the quantile-quantile plot construction
 as described in Section \ref{sec.properties}.
For example,  Figure \ref{quadraticQQQ} shows the QQ plot for the predictor $x$ when
 $q=0.1$: on the left hand side, we fitted a linear model, $y\sim x$,
 and on the right hand side 
the fitted model is $y\sim x  + x^2$ (the true model). The QQ plot suggests that, for the 10th percentile, 
 the linear model is inadequate,
while the quadratic one fits very well. The goodness of fit, as defined in Section \ref{sec.properties}, is
is $G=1,135$ for the linear model, and $G=828$ for the quadratic model, again providing evidence that in this case a quadratic model
provides a better fit.
  
 \begin{figure}[h!]
  \begin{center}
 \includegraphics[trim={1cm 4cm 4cm 1cm},clip,scale=0.6]{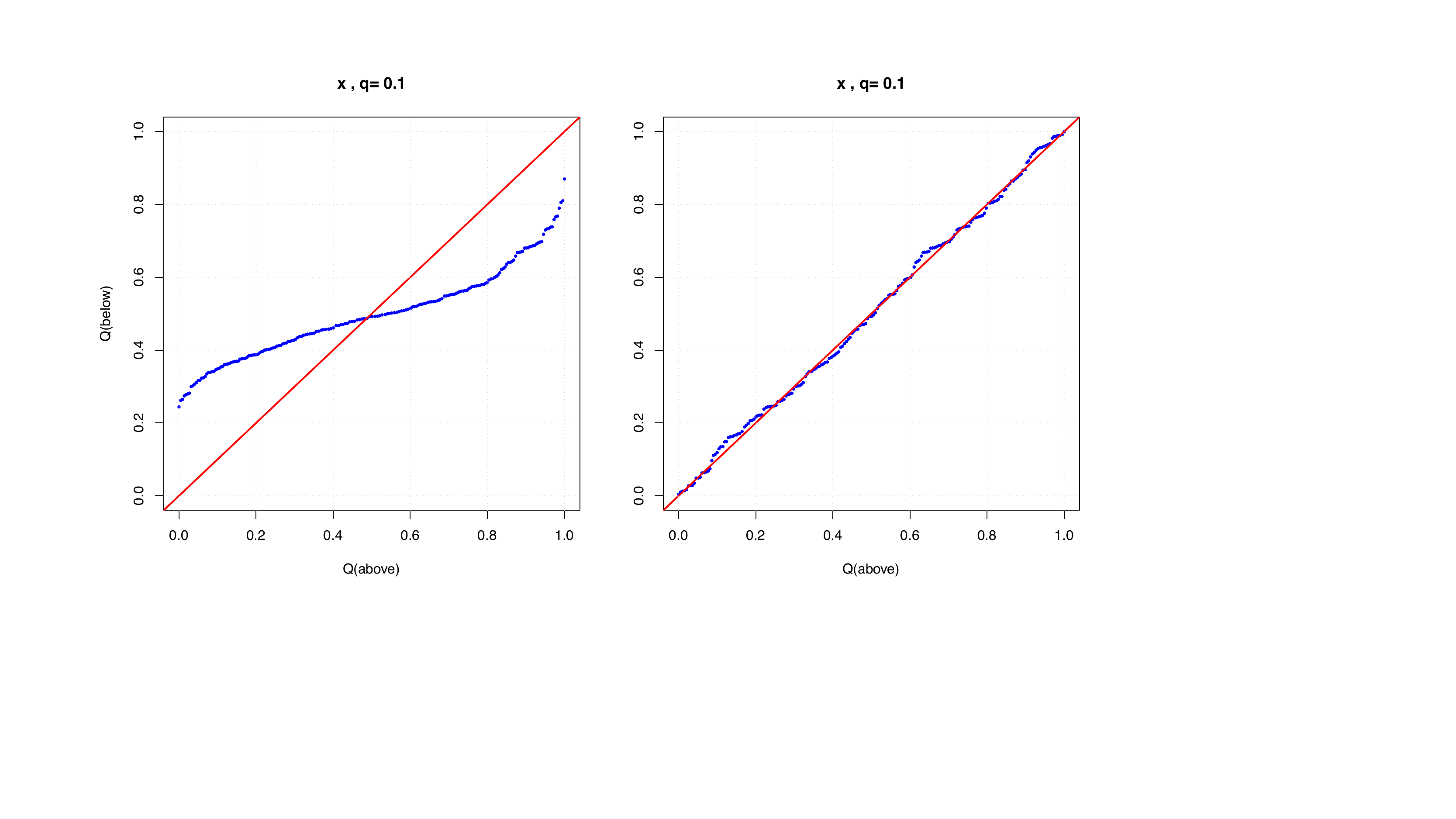}
\end{center}
\caption{QR simulation - the true model is $y \sim N(6x^2+x+120, (0.2+x)^2)$. Showing
quantile-quantile plots for the linear predictor $x$ when fitting a linear  (left) and the correct, quadratic model (right),
for $q=0.1$.}\label{quadraticQQQ}
 \end{figure}

It is possible that a model would fit well for certain quantiles, but not for others. The final example 
in this section demonstrates this point and shows an effective way to visualize multiple QQ-plots, when fitting the
same quantile regression model for different $q$s. 
We simulated 10,000 points from a bivariate uniform distribution on $[0,1]\times[0,1]$ as our two predictors, $x_1$ and
$x_2$, and generated the response using the interaction of the predictors, so that
 $y\sim N(4x_1x_2, (0.1+0.2x_1)^2)$ (simulation 24 in the Appendix).
For each $q\in(0.05,0.1,\ldots,0.9,0.95)$ we fitted two models - one additive, $y\sim x_1+x_2$, 
and one with an interaction term, $y\sim x_1+x_2+x_1x_2$.
From each fitted model we obtained the theoretical and empirical quantiles, $Q_{q,t}(x_1)$,  and $Q_{q,e}(x_1)$,
respectively, with respect to $x_1$. 
Recall that under the correct model the plot of the theoretical versus empirical quantiles should lie close to the
$45^{\circ}$ line. So, for a fixed $q$, and some $\xi_1$ in the range of $x_1$ we define
$r_q(\xi_1)=n_{q,e}(\xi_1)/n_{q,t}(\xi_1)$ where $n_{q,e}$ and $n_{q,t}$ are the numbers of 
empirical and theoretical quantiles that are smaller than $\xi_1$.
An adequate model gives $r_q(\xi_{1})\approx 1$ for each value of $\xi_1$.
For  each $q$ we use $L$ (e.g., 20) equally spaced values in the range of $x_1$, denoted by $\xi_{1j}$, 
and obtain $r_q(\xi_{1j})$.  We plot an array of rectangles with colors corresponding to the values of 
$r_q(\xi_{1j})$, so that the columns in the array correspond to the quantiles, and the rows to
$\xi_{1j}$.
This yields a heatmap, as depicted in Figure \ref{QxQy}, to which refer as a `flat' QQ plot, since for each $q$
we convert the two-dimensional QQ plot to a single column in the heatmap.
Figure \ref{QxQy}-A shows an ideal `flat' QQ-plot. The interaction model fits very well for each $q$.
In contrast, Figure \ref{QxQy}-B shows that although the additive model suggests 
a good fit for values around $q=0.5$, it is inadequate for most $q$'s.

 \begin{figure}
 \advance\leftskip-2cm
  \includegraphics[width=1.25\textwidth]{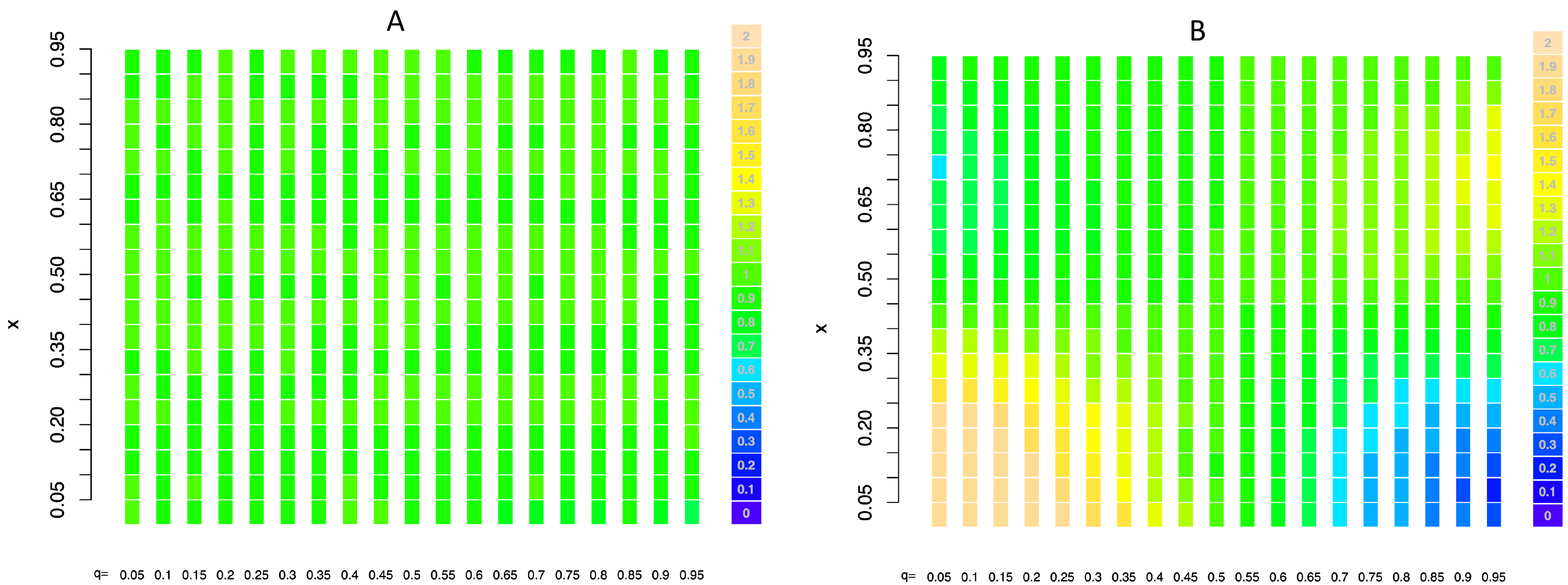}
 \advance\rightskip1cm
\caption{a `flat' QQ plot for $q\in(0.05,0.1,\ldots,0.9,0.95)$, when the true model is
$y\sim N(4x_1x_2, (0.1+0.2x_1)^2)$. A: fitting the correct (interaction) model. B. Fitting an additive (incorrect)
model, $y\sim x_1+x_2$.}\label{QxQy}
 \end{figure}

 Finally, we also simulated data from mixed models. In simulation 25 we generated
 a random samples of $100$ independent subjects, and each subject was observed at four
 time points. Within subject the observations were correlated.
That is, we used 
 a variance component model, where $y_{it}=2 + x_{it} + u_i + \epsilon_{it}$, and the random effect, $u_i$,
 was distributed as $N(0,0.5^2)$, and the random errors as $ \epsilon_{it}\sim N(0, 0.1^2)$. 

To obtain confidence intervals for the parameter estimates we used a bootstrap approach and fitted
the QR model 1,000 times for each $q=\{0.1,\ldots,0.9\}$, each time drawing a random sample (with replacement) from the subjects.
The parameter estimates and the bootstrap standard errors were very accurate for all deciles: 
the average bias for $\beta_1$ was 0.0008, and the 95\% coverage 
probability was $0.95\pm 0.005$.

\section{Case Studies}\label{sec.casestudy}

\subsection{BMI and Microbiome Data}\label{BMI}
An operational taxonomic unit (OTU) is an operational definition used to classify groups of closely related organisms.  OTUs have been the most commonly used units of microbial diversity, especially when analyzing small subunit 16S rRNA marker gene sequence datasets. In this case study we illustrate a quantile regression approach to the analysis of microbiome compositional covariates using BMI and OTU data from \cite{wu:2011}.    \cite{bar2018scalable} applied SEMMS (Scalable EMpirical Bayes Model Selection) with BMI as a continuous response with compositional OTU covariates, using the normal model and obtained the same four genera reported by the LASSO-based variable selection method in \cite{Lin13082014} (\emph{Acidaminococcus, Alistipes, Allisonella, and Clostridium}).  \cite{bar2018scalable} also demonstrated the application of SEMMS using a categorical version of BMI with 3 levels: normal [18.5, 25) (n=30), overweight [25.5, 30) (n=25), and obese $\ge$30 (n=10) and found six bacteria associated with obesity relative to normal (\emph{Acidaminococcus, Alistipes, Allisonella, Butyricimonas, Clostridium, and Oxalobacter}) and five with overweight relative to normal (\emph{Anaerofilum, Faecalibacterium, Oscillibacter, Turicibacter, and Veillonella}). A logistic regression approach using SEMMS for the obese group (n=10) when the baseline consisted of the non-obese subjects (BMI$<$30, n=86) found five genera (\emph{Acidaminococcus, Alistipes, Allisonella, Dialister, and Dorea}). The fact that \emph{Acidaminococcus, Alistipes, Allisonella, and Clostridium} are a subset of the genera found to be associated with obese individuals but have no overlap with the ones associated with overweight suggest that different BMI levels are associated with different bacteria, and thus hint that quantile regression will give a more interpretable analysis than the conditional mean regression model.

Since microbiome abundance data is compositional, to apply our method directly, without changing the model to account for the sum constraint, we perform the log ratio transformation and replace the matrix $X$ with $X^K =\log(x_{ij}/x_{iK})$, where the $x_{iK}$ are replications of a reference bacteria genera. The data contains many zero counts, so \cite{bar2018scalable, Lin13082014} replace them with 0.5 before converting the data to be in compositional form.  As in \cite{Lin13082014,bar2018scalable} we use a subset of 45 bacteria which had non-zero counts in at least 10\% of the samples ($N=96$). The omitted genera have minimal contribution to the overall distribution of the proportions.

\begin{figure}[b!]
  \begin{center}
 \includegraphics[trim={0 1cm 0 2cm},clip,scale=0.3]{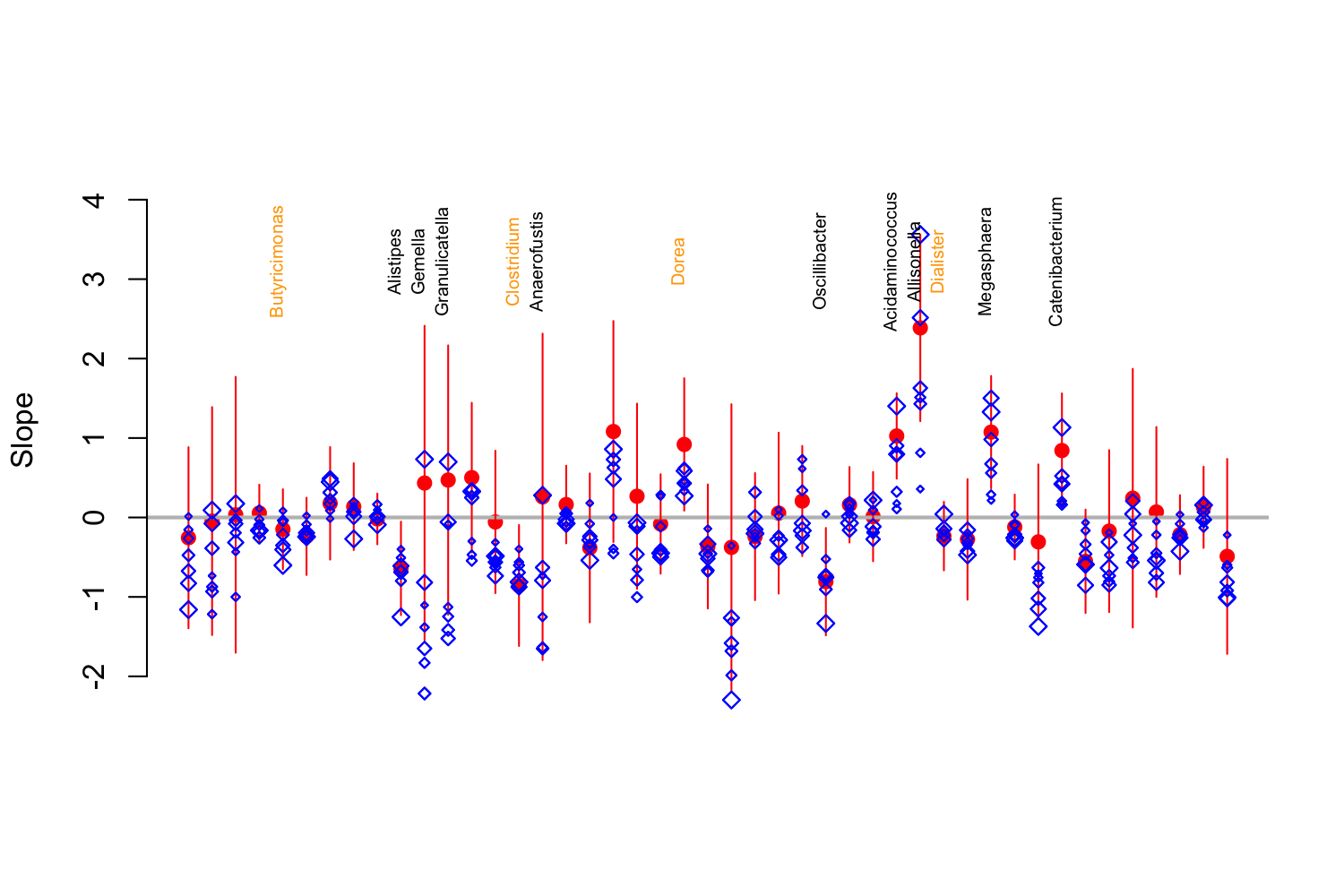}
\end{center}
\caption{Mean-model estimates and their confidence intervals in red, and the quantile regression coefficients for $q=0.2, \ldots, 0.8$ in blue (diamonds, where the increased size corresponds to larger $q$).}\label{BMIplot1}
 \end{figure}

\begin{figure}[t!]
  \begin{center}
 \includegraphics[trim={0 2.5cm 0 2cm},clip,scale=1.1]{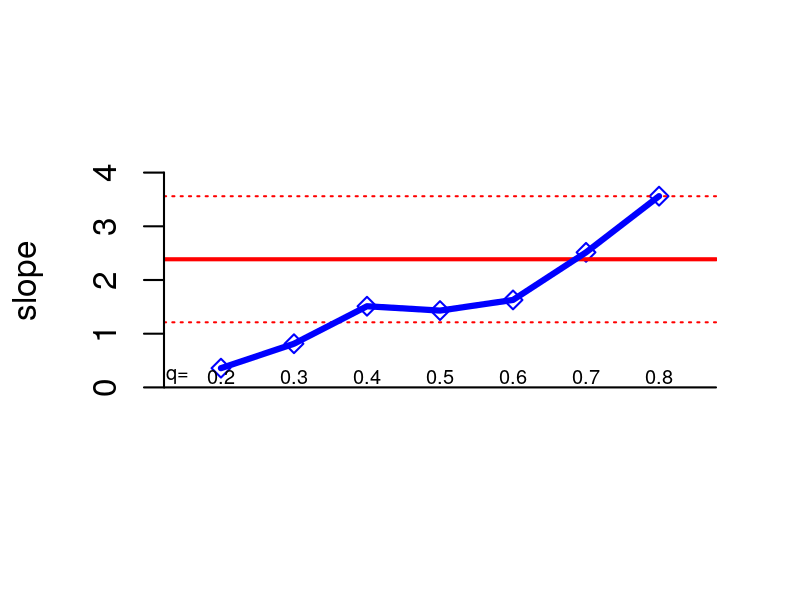}
 \includegraphics[trim={0 3cm 0 2cm},clip,scale=.4]{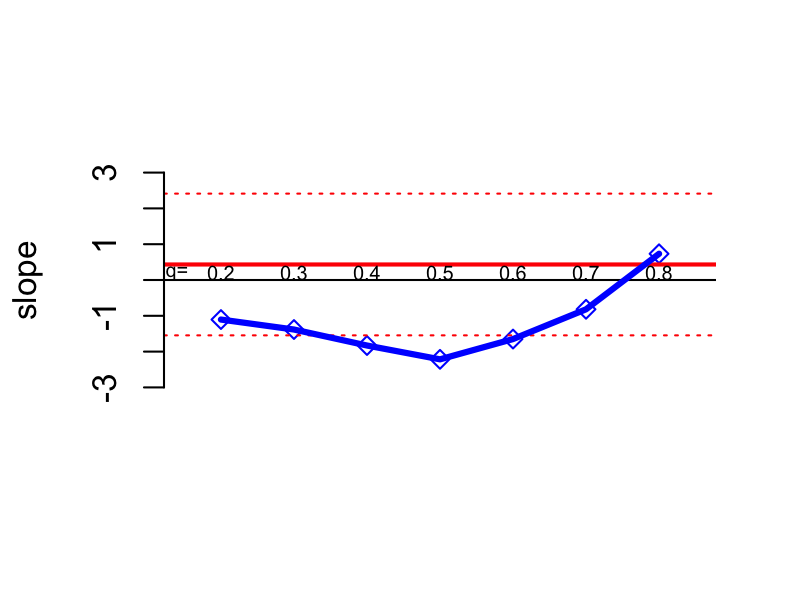}
\end{center}
\caption{Mean-model estimates and their confidence intervals in red and quantile regression coefficients for $q=0.2, \ldots, 0.8$ for \emph{Allisonella} (top) and \emph{Gemella} (bottom).}\label{BMIplot2}
 \end{figure}

Figure \ref{BMIplot1} shows the mean-model estimates and their confidence intervals in red, and the quantile regression coefficients for $q=0.2, \ldots, 0.8$ in blue (diamonds, where the size corresponds to $q$). This is done one OTU at a time with no adjustment for multiple testing. Most confidence intervals are quite large, and contain all the quantile regression estimates.  The quantile regression analysis leads to a number of interesting points.  For example in Figure \ref{BMIplot2},  \emph{Allisonella} has a strong mean effect ($\sim 2.5$) but even stronger quantile effect for $q=0.8$. Notice also that \emph{Gemella, Granulicatella, Anaerofustis} have a small positive mean-effect estimate, but the estimates for the lower quantiles are negative and are outside the mean confidence interval.   For some OTUs the effect for all quantiles has the same sign, maybe suggesting that these OTUs should be considered as significant even if the adjusted p-value for the mean effect is not sufficiently small. For example, \emph{Acidaminococcus, Megasphaera, Catenibacterium} (all positive) and \emph{Alistipes, Clostridium, Oscillibacter} (all negative)  have a fairly small spread of quantile regression estimates, perhaps suggesting that they are in fact significant, even if their p-value after adjustment is not small enough.

The quantile regression results are consistent with the previous findings in the microbiome literature.  The relative proportion of phylum \emph{Bacteroidetes} (containing the \emph{Alistipes} genus) to phylum \emph{Firmicutes} (all other genera found) is lower in obese mice and humans than in lean subjects \cite{ley:2006}.  In \cite{andoh2016comparison}
\emph{Gemella, Granulicatella} were both found to be increased bacteria genera in obese people but not as a differential between lean and obese people, this is reflected in the quadratic behavior seen for \emph{Gemella} in Figure \ref{BMIplot2}.  Furthermore, the family \emph{Veillonellaceae} (containing the \emph{Acidaminococcus} and \emph{Allisonella} genera) are positively correlated to BMI \cite{wu:2011}. The two additional genera found beyond \cite{Lin13082014}, \emph{Catenibacterium} and \emph{Megamonas}, have been found to be positively associated with BMI. \cite{chiu2014} identify \emph{Megamonas} as a genus that differentiates between low and high BMI in a Taiwanese population and \cite{turnbaugh2009} found that the \emph{Catenibacterium} genus increased diet-induced dysbiosis for high-fat and high-sugar diets in humanized gnotobiotic mice.

\subsection{Daily Temperature Data}\label{temp}
We demonstrate fitting a quantile regression mixed-model using our approach to daily temperature
data from four U.S. cities (Chicago, Ithaca, Las Vegas, and San Francisco). Minimum and maximum
temperatures (in Fahrenheit) from January 1, 1960 to December 31, 2010 were obtained from the
National Oceanic and Atmospheric Administration website (https://www.noaa.gov/).
We fit a  model which allows for the temperature characteristics to change linearly
over time, while accounting for the cyclical nature of
temperatures. To allow for extra variability in the data in addition
to between-days variability, we include a random intercept for each month of the year.
Let $y_{ijk}$ denote the daily maximum (minimum) temperature in month
$j$ and day $k$ of year $i$ and define
$S_k=\sin(2\pi(k-11)/365-\pi/2)$, for $k=1,\ldots,365$ (or $366$), to
represent the annual sinusoidal variation of daily temperatures,
with a minimum at (approximately) the  shortest
day of the year in the northern hemisphere,  December 21, and a maximum at around June 21.
Our main objective is to estimate linear trends in temperatures over
the 50 year period. For each city we fit six
models with linear predictors of the form 
$$ \eta_{ijk}=\mu+\beta_1i+\beta_2S_k + v_j \,.$$
for the 10-th percentile, the mean, or the 90-th percentile of the daily minimum or maximum
temperature, and where $v_j$, $j=1,\ldots,12$ is a (random) effect of month $j$. 
The \texttt{lmer} function \cite{lme4:2015} was used to
fit the parameters in each iteration of the EM algorithm (in the M-step.)

The estimates of $\beta_1$ from each model are summarized in Table \ref{climatechange} for all four cities.
All cells with total temperature increase of more than 1 degree are statistically significant.
For example,  in San Francisco the minimum daily temperatures rose on average
by only 0.1 degrees from 1960 to 2010, but
the  90-th percentile  of the maximum daily temperature
 increased by 4.7 degrees. The mean maximum daily
temperatures increased by 3.4 degrees in those 50 years.
In Las Vegas, on the other hand, the 90-th percentile of the maximum daily temperature increased by 1 degree,
but the 10-th percentile of the minimum daily temperature  increased by almost 9 degrees.
In both Chicago and Las Vegas, the largest increase for both minimum and maximum daily temperatures has been
to the 10-th percentile, which means that the 10\% `cooler than usual for the time of year' days were quite warmer in 2010
than in 1960.

\begin{table}[ht]
\centering
\begin{tabular}{|l|ccc|ccc|ccc|ccc|}
  \hline
  & \multicolumn{3}{c|}{Chicago} & \multicolumn{3}{c|}{Ithaca} & \multicolumn{3}{c|}{Las Vegas} & \multicolumn{3}{c|}{San Francisco} \\
  & \multicolumn{3}{c|}{Illinois} & \multicolumn{3}{c|}{New York} & \multicolumn{3}{c|}{Nevada} & \multicolumn{3}{c|}{California} \\
  & 10\% & M & 90\% & 10\% & M & 90\% & 10\% & M & 90\% & 10\% & M & 90\% \\
  \hline
min temp.  & 4.3 & 2.9 & 1.6 & 2.1 & 1.8 & 1.3 & 8.7 & 7.8 & 6.9 & 0.1 & 0.1 & 0.2 \\
max temp.  & 1.9 & 1.2 & 0.3 & 1.1 & 1.2 & 1.5 & 2.0 & 1.5 & 1.0 & 2.5 & 3.4 & 4.7 \\
   \hline
\end{tabular}
\caption{The total change in temperature from 1960-2009 in four U.S. cities. In each city we estimate the
change in the minimum and maximum temperatures in terms of the 10-th percentile, the mean, and the 90-th
percentile.}\label{climatechange}
\end{table}

The fitted lines for the quantile regression models with $q=0.1, 0,9$ are shown for the Las Vegas data for 1960-64 and 2006-10
in Figure \ref{LV}. The sharp increase in minimum daily temperatures of approximately 8 degrees
between 1960-4 and 2006-10
can be seen very clearly (top).
Figure \ref{tempfitted} shows the daily minimum and maximum temperatures for Chicago and San Francisco,
 from 1985-1990, as well as the fitted lines from our quantile regression for $q=0.1$ (blue) and $q=0.9$ (red).
 The plot shows that although both cities have cyclical patterns, their shapes are quite different.
 The mixed model in which the cyclical fixed effect
 is dominant and the random month-by-month intercept allows for deviations from the sinusoidal pattern
 seems to fit both patterns quite well. For Chicago, the
 sinusoidal shape is very clear, whereas in San Francisco the pattern is skewed with a sharp decline in temperatures
 in the fall and a slower incline starting in January. It looks like late-spring and summer in San Francisco
 are a lot more likely to have maximum temperature which are much higher than the mean-model
 predicts (black curve in the bottom-right plot), and the 90-th percentile seems much more informative.

\begin{figure}
  \begin{center}
 \includegraphics[scale=0.8,angle=270]{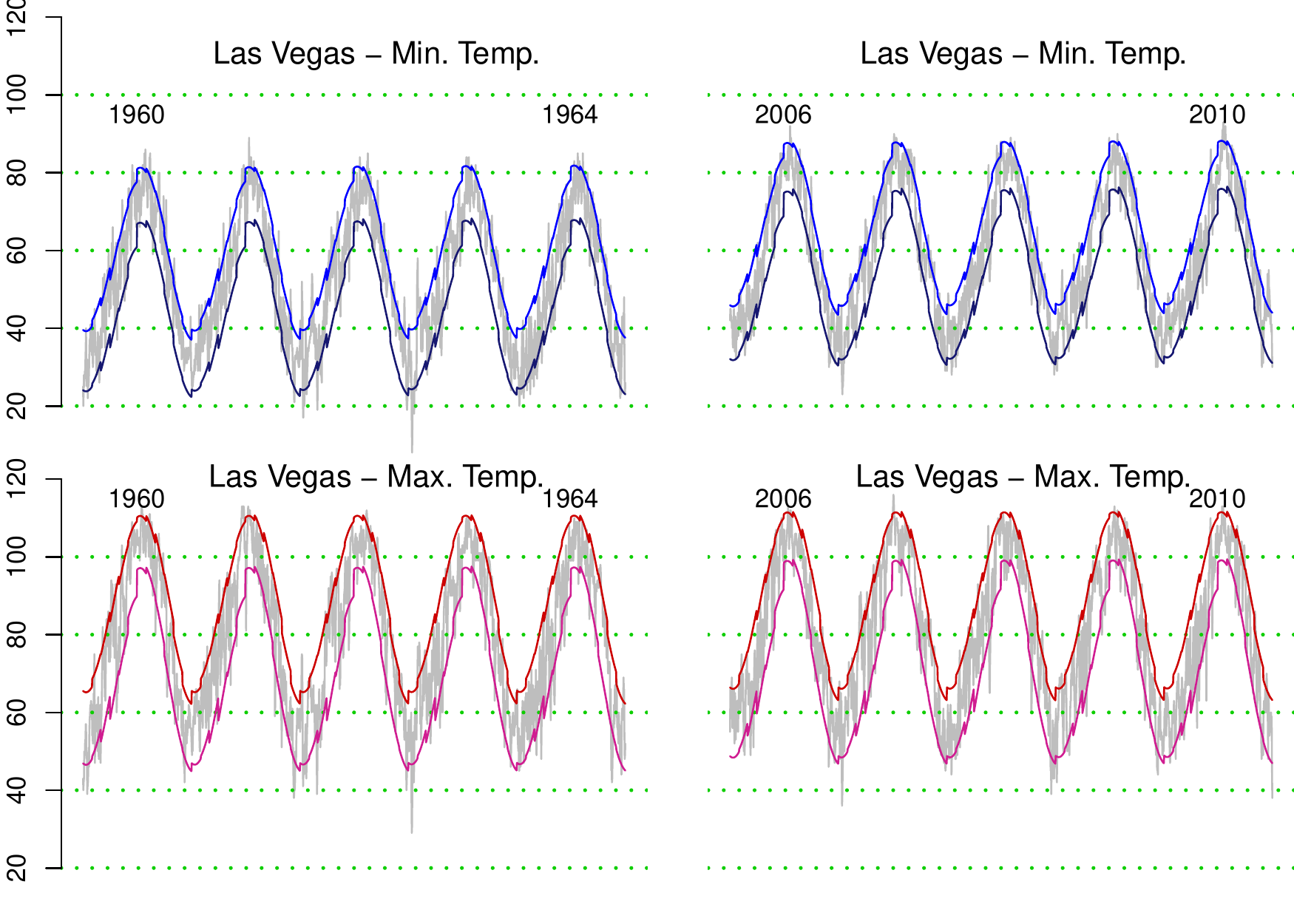}
\end{center}
\caption{The fitted lines for the quantile regression models for Las Vegas for 1960-1964 (left) and 2006-2010 (right) with $q=0.1$  and  $q=0.9$. The actual observations are shown in grey. Top: minimum daily temperature, bottom: maximum
daily temperature.}\label{LV}
 \end{figure}

\begin{figure}
  \begin{center}
 \includegraphics[scale=0.8,angle=270]{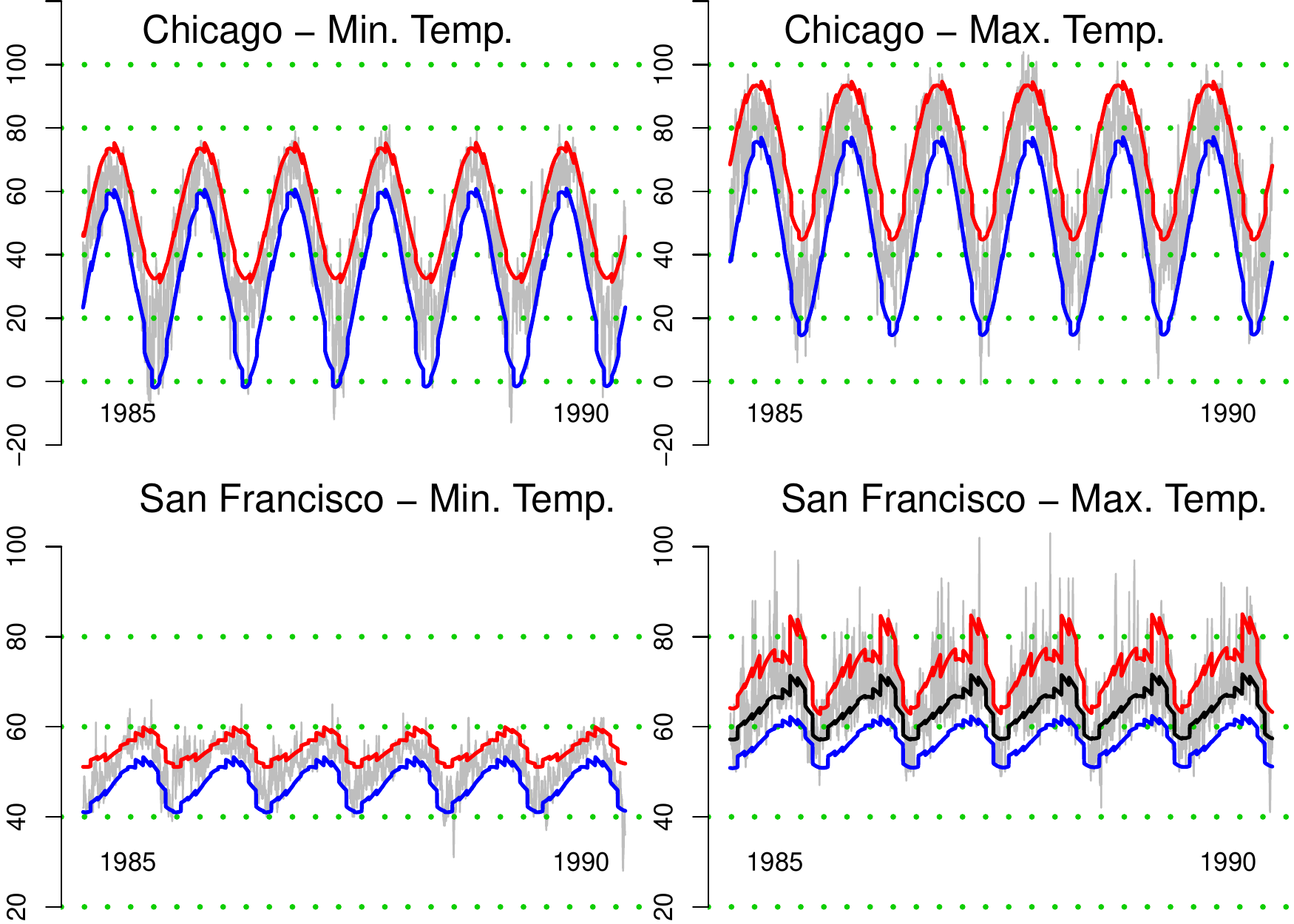}
\end{center}
\caption{Daily minimum and maximum temperatures for Chicago and San Francisco,
 from 1985-1990. Fitted lines from our quantile regression for $q=0.1$ in blue and $q=0.9$ in red.
 For max. temperatures in San Francisco we show the fitted curve from the mean model in black.}\label{tempfitted}
 \end{figure}

 To check the model fit, we produced QQ plots as demonstrated in the previous section.
 For Chicago, Ithaca, and Las Vegas, the plot for each predictor and each quantile (0.1 and 0.9) was
 along the diagonal, for both the minimum and maximum temperatures. See, for example, the QQ plot for the 
 year predictor for the maximum temperatures in Las Vegas in Figure \ref{QQtemps} on the left.
 However, for San Francisco, the QQ plots for the year variable exhibited a slightly curved pattern 
 (Figure \ref{QQtemps}, right-hand side). Looking at the time-series plot (not shown here) 
 we see that this could be due to a period of time (1984-1997) 
in which the percentage of especially high
 maximum temperatures in a single year was higher than usual.
 
 \begin{figure}
  \begin{center}
 \includegraphics[trim={1cm 8cm 8cm 1cm},clip,scale=0.7]{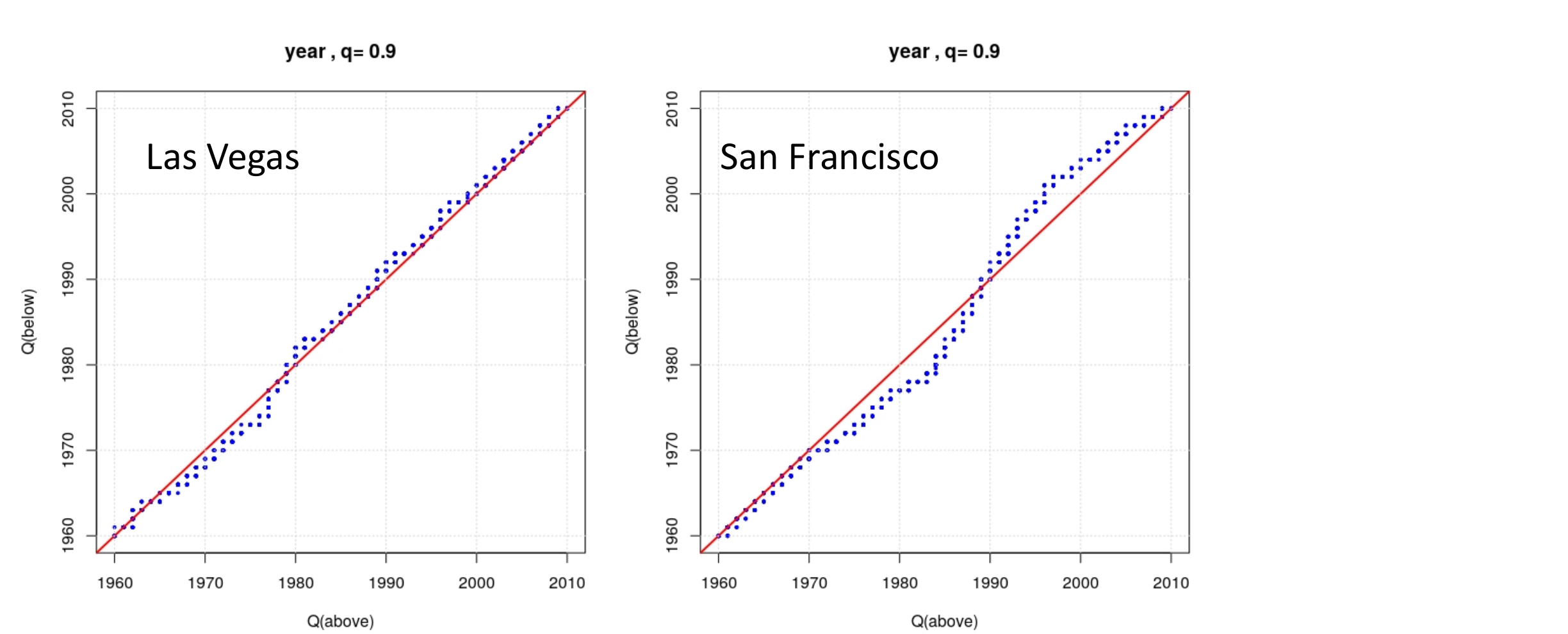}
\end{center}
\caption{Quantile-quantile plot for the maximum temperatures in Las Vegas (left) and San Francisco (right), corresponding to the `year' predictor, when fitting the quantile regression with $q=0.9$}\label{QQtemps}
 \end{figure}
 
 Figure \ref{flatQQtemps} shows the `flat' QQ plot for the maximum temperature in San Francisco,
 for $q=0.1,0.2,\ldots,0.9$, with respect to the year predictor. Using the notation from the previous section,
 the range of $r_q(\xi_{year,j})$ is $0.88-1.26$. Although most of the values are close to 1, the model 
 does not fit the data for San Francisco as well as for the other three cities (not shown here). Note that 
 the goodness of fit of the model varies by the selected quantile, $q$. In particular, the worst fit is for the median,
 where, around 1970, the predicted values are 1.26 higher than the observed median temperatures.
 
 \begin{figure}
  \begin{center}
 \includegraphics[scale=0.7]{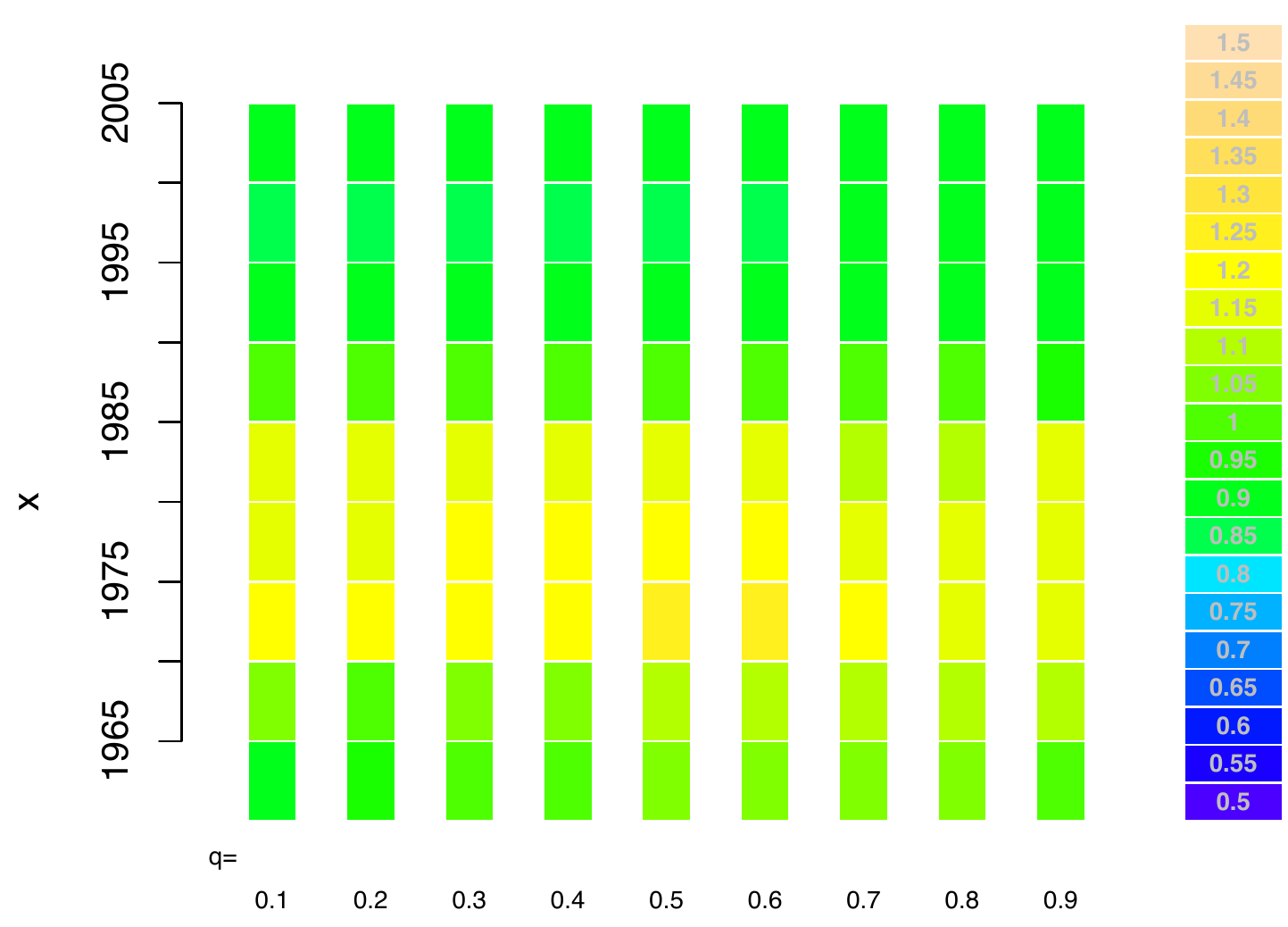}
\end{center}
\caption{A `flat' Quantile-quantile plot for the maximum temperatures in San Francisco with respect
 to the `year' predictor, for $q=0.1,0.2,\ldots,0.9$}\label{flatQQtemps}
 \end{figure}
 
\subsection{Frailty QR Model -- Emergency Department Data}
The National Center for Health Statistics, Centers for Disease Control and Prevention
conducts annual surveys to measure nationwide health care utilization.
We use the 2006 NHAMCS (National Hospital Ambulatory Medical Care Survey) data to demonstrate
fitting quantile regression models with random effects to survival-type data.
In this case, the length of visit (LOV, given in minutes) while in a hospital emergency department (ED)
is considered the `survival' time, and we define the normalized response $y=\log_{60}(LOV+1)$.
We filter the data and remove hospitals with fewer than 70 visits, which results in 21,262 ED visit records from 230
hospitals.

First, we fit a fixed-effect model which includes nine predictors of interest:
Sex, Race (white, black, other), Age (standardized to a $[0,1]$ range), Region (northeast, midwest,
south, and west), Metropolitan area (yes/no), Payment Type (Private, Government/Employer, Self, and Other),
Arrival Time (8AM-8PM or 8PM-8AM), the Day of the Week, and a binary variable (Recent Visit) to indicate
whether the patient has been discharged from a hospital in the last 7 days, or from an ED within the last 72 hours.
We fit a QR (fixed-effect) model for $q=0.025, 0.05, 0.075,\ldots,0.975$.
Then, we fit the frailty model for the same quantiles using the same nine predictors, plus
the Hospital as a random effect, and check whether the coefficients of the fixed effects change
when we account for the correlation between patients within a hospital.

\begin{figure}[b!]
  \begin{center}
 \includegraphics[trim={1cm 6cm 4cm 3.5cm},clip,scale=0.54]{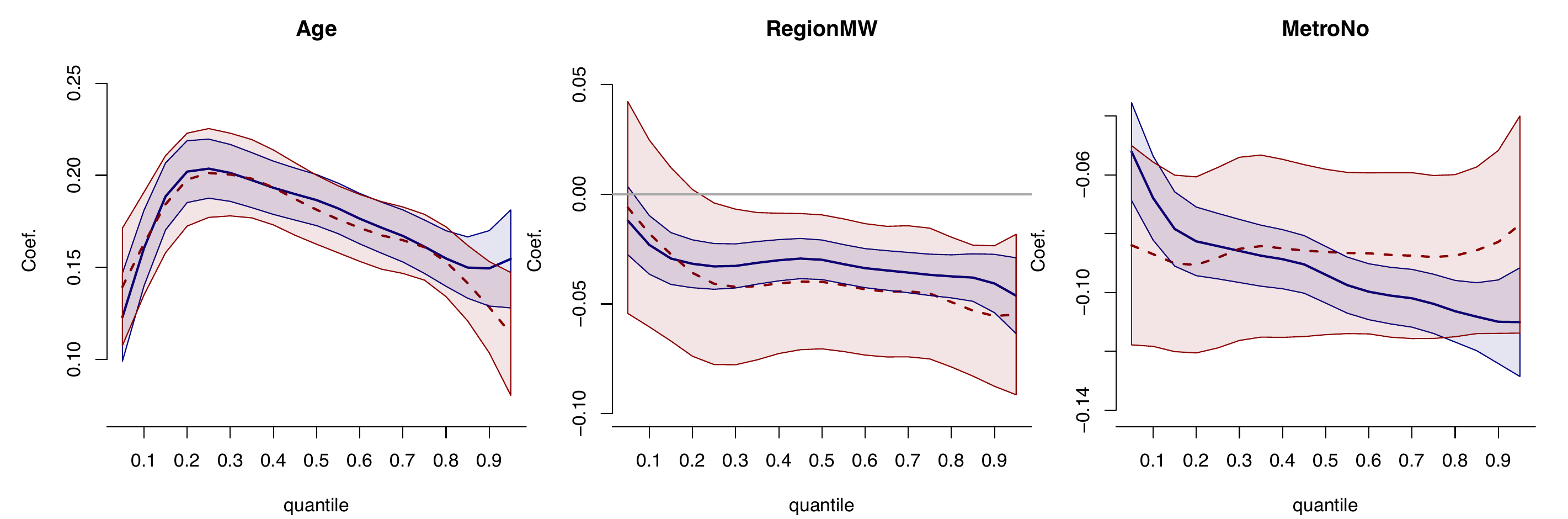}
\end{center}
\caption{ED data -- quantile regression coefficients with 95\% point-wise confidence intervals, 
from left to right, for Age, Region (=midwest, where northeast is the
 baseline), and Metro (=No, where Yes is the baseline). The blue, solid line is the for the fixed effect model
 coefficients, and the red, dashed line is for the mixed model coefficients.}\label{ERf3}
 \end{figure}
 
 Figure \ref{ERf3} shows the regression coefficients, from left to right, for Age, Region (=midwest, where northeast is the
 baseline), and Metro (=No, where Yes is the baseline). For these variables, the coefficients are remain approximately the
 same when we add the Hospital random effect. Age has a positive effect for all quantiles, and the smallest difference in LOV
 due to age is for the patients who are discharged quickly from the ED, or those who stay at the ED the longest.
 Generally, ED patients in the northeast have longer visits than the ones in the midwest (except for those
 who are discharged quickly). Similarly, ED patients in a metropolitan area stay significantly longer than in
 rural area hospitals.
 The Arrival Time and Day of Week variables are also similar whether we include the hospital random effect or not (no
 shown here.) Arriving at an ED on Monday results in a longer wait compared with Sunday (except for patients who 
 are discharged quickly), but the difference between other weekdays and the weekend is less significant for most
 quantiles. Arriving at night means a longer stay only for the patients who end up staying the longest (the upper 
 20-th percentile in the fixed effect model but no significant difference in the mixed model), but among patients
 who are discharged relatively quickly ($q\le0.2$) arriving at night actually corresponds to a shorter stay, as compared
 with day-time arrival.

Figure \ref{ERf2} shows that accounting for within-hospital correlation yields
very different results compared with the fixed-effect model. According to the fixed effect model
one might conclude that there is a significant difference between blacks and whites (left panel)
and between people with private health insurance versus people whose medical bill is `Other' (namely, 
not Private, Medicare, Medicaid/SCHIP, Worker's Compensation, or Self pay), for patients who are 
not discharged quickly, i.e., for $q>0.2$. Similarly, from the fixed effect model one might conclude that among
those who are not discharged quickly, a recent visit to a hospital or an ED will result in longer stays (for $q>0.35$).
However, the results from the mixed model suggest that these differences can be explained by variation between 
hospitals, while, within a hospital there is no significant difference in the length of visit
between black and white people, or between people with private insurance and people who `other' pay.

 \begin{figure}[b!]
  \begin{center}
 \includegraphics[trim={3cm 6cm 4cm 3.5cm},clip,scale=0.6]{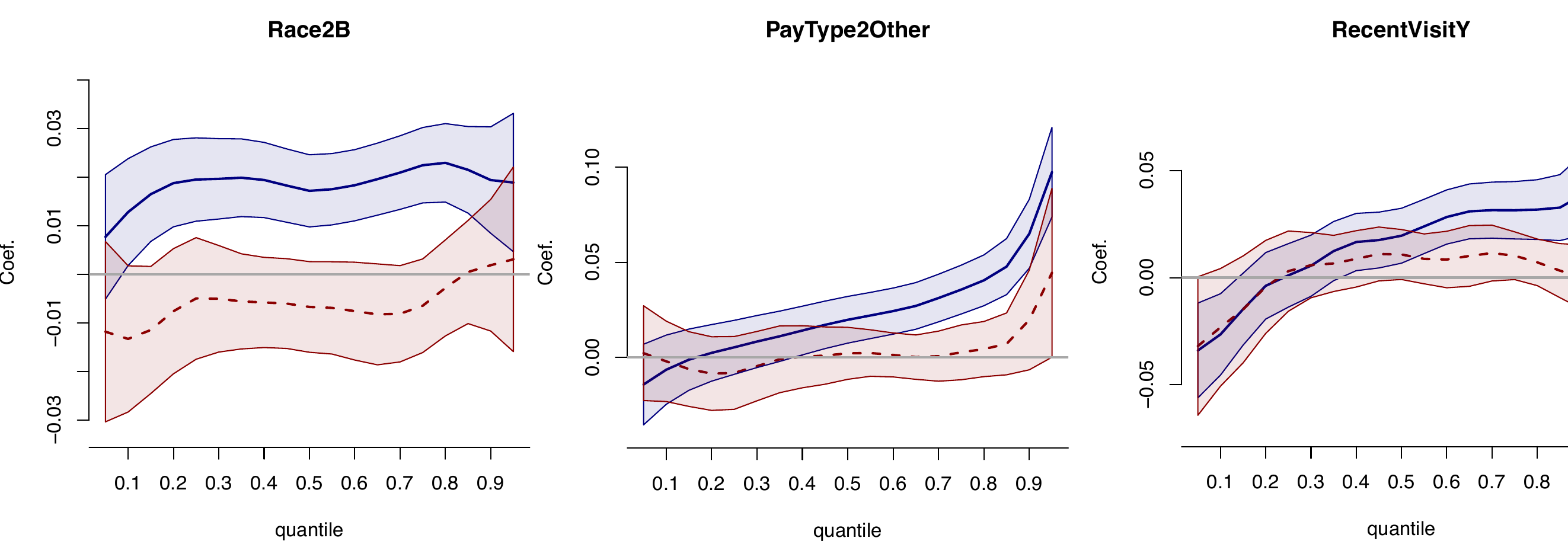}
\end{center}
\caption{ED data -- quantile regression coefficients with 95\% point-wise confidence intervals,
for Race=Black (left, with White as baseline),  Pay Type (center, for `Other' source of payment versus `Private'),
and Recent Visit (for Yes, versus the No reference.)
The blue, solid line is the for the fixed effect model
 coefficients, and the red, dashed line is for the mixed model coefficients.}\label{ERf2}
 \end{figure}

\section{Discussion}\label{sec.discussion}
The EM formulation introduced in this paper for fitting quantile
regression models offers a few advantages over existing methods
stemming from the fact that the M-step involves a mean regression via
weighted least squares. This formulation suggests an approach to fitting
mixed effects quantile regression models that can be justified as
a GAM algorithm. We see several possibilities for 
interesting applications and extensions, such as quantile regression
with autoregressive (or ARMA) errors, and spatial models that use the 
mixed models framework to specify the correlation structure. 
Another extension is to variable selection for the `large p' scenario 
in which the variables are selected based
on their association with quantiles, rather than a mean response, by modifying the 
 methodology developed in \cite{bar2018scalable}.

\appendix
\section{Appendix}\label{sec.appendix}

\subsection{Proofs}
\textbf{Proof of Proposition \ref{propPLik}:} \\
First note that the check loss function (\ref{checkloss}) can be written as
$$\rho_q(u_i)=\frac{1}{2}\left[|u_i|+(2q-1)u_i\right]\,, \mbox{ so we can write:}$$
$$e^{-2\rho_q(u_i)} =e^{-(2q-1)u_i}e^{-|u_i|}\,.$$
The marginal distribution of $u$ is denoted by $h(u)$, where
\begin{eqnarray*}
 h(u) &=& \int_0^\infty\frac{1}{\sqrt{2\pi\lambda}} \exp\left[- \frac{(u-(1-2q)\lambda)^2}{2\lambda}\right]
 \cdot2q(1-q) \exp\left[- 2q(1-q)\lambda\right] d\lambda\\
 &=& 2q(1-q)\int_0^\infty\frac{1}{\sqrt{2\pi\lambda}} \exp\left[- \frac{u^2-2u(1-2q)\lambda +
 (1-2q)^2\lambda^2 + 4q(1-q)\lambda^2}
 {2\lambda}\right] d\lambda\\
&=& 2q(1-q)\int_0^\infty\frac{1}{\sqrt{2\pi\lambda}} \exp\left[- \frac{u^2-2u\lambda+\lambda^2 (1-4q+4q^2+4q-4q^2) +4qu\lambda
 }
 {2\lambda}\right] d\lambda\\
&=& 2q(1-q)\int_0^\infty\frac{1}{\sqrt{2\pi\lambda}} \exp\left[- \frac{u^2-2u\lambda+\lambda^2}
 {2\lambda} - 2qu\right] d\lambda\\
 &=& 2q(1-q)\int_0^\infty\frac{1}{\sqrt{2\pi\lambda}} \exp\left[- \frac{u^2/\lambda-2u+\lambda}{2}
 - 2qu\right] d\lambda\\
 &=& 2q(1-q)\int_0^\infty\frac{1}{\sqrt{2\pi\lambda}} \exp\left[-\frac{u^2/\lambda+\lambda}{2}
 - (2q-1)u)\right] d\lambda\\
 &=& 2q(1-q)e^{-(2q-1)u}\int_0^\infty\frac{1}{\sqrt{2\pi\lambda}}\exp\left\{-\frac{1}{2}(\lambda+u^2/\lambda)\right\}d\lambda
 \end{eqnarray*}
From \cite{pols:scot:2012} and formula 12 on page 368 of \cite{gradshteyn2014table} (with $y=x/2$ and $\nu=.5$) we have that
\begin{equation*}
e^{-|u|}=\int_0^\infty\frac{1}{\sqrt{2\pi\lambda}}\exp\left\{-\frac{1}{2}(\lambda+u^2/\lambda)\right\}d\lambda\,,
\end{equation*}
It follows that, 
\begin{eqnarray*}
 h(u) &=& 2q(1-q)\exp\{-[(2q-1)u+|u|]\} = 2q(1-q)\exp[-2\rho_q(u)]\,.
 \end{eqnarray*}

\noindent\textbf{Proof of Proposition \ref{propUncorr}:} \\
We denoted  the scaled, binary-valued QR
residuals by $\mathbf{c}=sgn(Y-X^T\bm\beta_q) - (1-2q)\mathbf{1}$. We
found the WLS
solution from the M-step:
$$\bm\beta_q = (X^T\Lambda^{-1}X)^{-1}X^T\Lambda^{-1}\left[Y-(1-2q)\bm\lambda\right].$$
Multiplying both sides by $(X^T\Lambda^{-1}X)^{-1}$ we get
$$(X^T\Lambda^{-1}X)\bm\beta_q = X^T\Lambda^{-1}\left[Y-(1-2q)\bm\lambda\right]$$
and rearranging terms we get
\begin{eqnarray*}
X^T\Lambda^{-1}(Y-X^T\bm\beta_q)&=&(1-2q)X^T\Lambda^{-1}\bm\lambda\\
&=&(1-2q)X^T\mathbf{1}\,.
\end{eqnarray*}
Expressing $Y-X^T\bm\beta_q$ as $\Lambda\cdot sgn(Y-X^T\bm\beta_q)$ we get
\begin{eqnarray*}
X^T\Lambda^{-1}[\Lambda\cdot sgn(Y-X^T\bm\beta_q)]&=&(1-2q)X^T\mathbf{1}
\end{eqnarray*}
so, $X^T\mathbf{c}=\mathbf{0}$.

\subsection{Conditions for consistency of the QR estimator}
Let the $q$th conditional quantile function of $Y|\mathbf{X}$ be $Q_Y(q|\mathbf{X})$.
Per \cite[Section 4.1.2]{koen:2005}, for the quantile regression estimator $\hat{\bm{\beta}}_q$
to be consistent, the following conditions have to hold:
\begin{enumerate}
\item There exists $d > 0$ such that
$$\underset{n\rightarrow\infty}{\lim\inf} \underset{\|u\|=1}{\inf} n^{-1}\sum I(|\mathbf{x}_i'u|<d)=0\,.$$
\item There exists $D > 0$ such that
$$\underset{n\rightarrow\infty}{\lim\sup} \underset{\|u\|=1}{\sup} n^{-1}\sum (\mathbf{x}_i'u)^2\le D\,.$$
\end{enumerate}

\subsection{Simulation Scenarios}
\begin{landscape}
\begin{table}[b!]
\begin{tabular}{|l|l|l|l|}
\hline
No. & Description & Model for the response, $y$ & Error distribution\\
\hline
1    & Intercept only & $3$ & $N(0,0.25^2)$\\
2-11 & Simple linear model & $5-x$ & $N(0,\sigma^2)$, $\sigma\in\{0.1,0.2,\ldots,1\}$\\
12   & Two predictors & $1-3x_1+2x_2$ & $N(0,0.1^2)$\\
13   & Five predictors & $1-3x_1+2x_2+2x_3-x_4-2x_5$ & $N(0,0.1^2)$\\
14   & s.d. increases linearly  & $3+2x$ & $N(0,(0.1+0.2x)^2)$\\
15   & s.d. increases linearly  & $5+x$ & $N(0,(0.1+0.5x)^2)$\\
16   & s.d. increases linearly & $3+0.5x$ & $N(0,(0.5+0.7x)^2)$\\
17   & Polynomially increasing s.d. & $1-2x$ & $N(0,(0.1+0.2x^3)^2)$\\
18   & Linearly decreasing s.d. & $7+3x$ & $N(0,(1-0.5x)^2)$\\
\hline
19   & Intercept only & $5$ & $LN(0,0.75)$\\
20   & Simple linear model & $3-x$ & $LN(0,0.75)$\\
21   & Five predictors & $1-3x_1+2x_2+2x_3-x_4-2x_5$ & $LN(0,0.75)$\\
22   & Linearly increasing (log) s.d. & $2-2x$ & $LN(0,0.5+0.5x)$\\
\hline
23   & Quadratic, increasing variance & $6x^2+x+120$ & $N(0, (0.2+x)^2)$\\
24   & Interaction, increasing variance  & $4x_1x_2$ & $N(0, (0.1+0.2x_1)^2)$\\
\hline
25   & Mixed model  & $2 + x + z'u$ & $u \sim N(0,0.5^2)$, $e\sim N(0, 0.1^2)$\\
\hline
\end{tabular}
\caption{Simulation scenarios. In simulations 1-11, 14-17, and 24 the predictors were drawn uniformly 
from  $(0,1)$. In simulation 12 $x_1\sim U(0,1)$ and $x_2\sim U(-3,3)$.
In simulations 13, 18-22 the predictors were drawn uniformly 
from  $(-1,1)$ and in simulation 23, from $(-5,5)$.
In simulation 25, $x_{it}\sim N(t/4, 0.1^2)$, for $t=1,2,3,4$.}\label{sims}
\end{table}

\end{landscape}

\begin{figure}
  \begin{center}
 \includegraphics[trim={0 0.5cm 2 1cm},clip, scale=0.7]{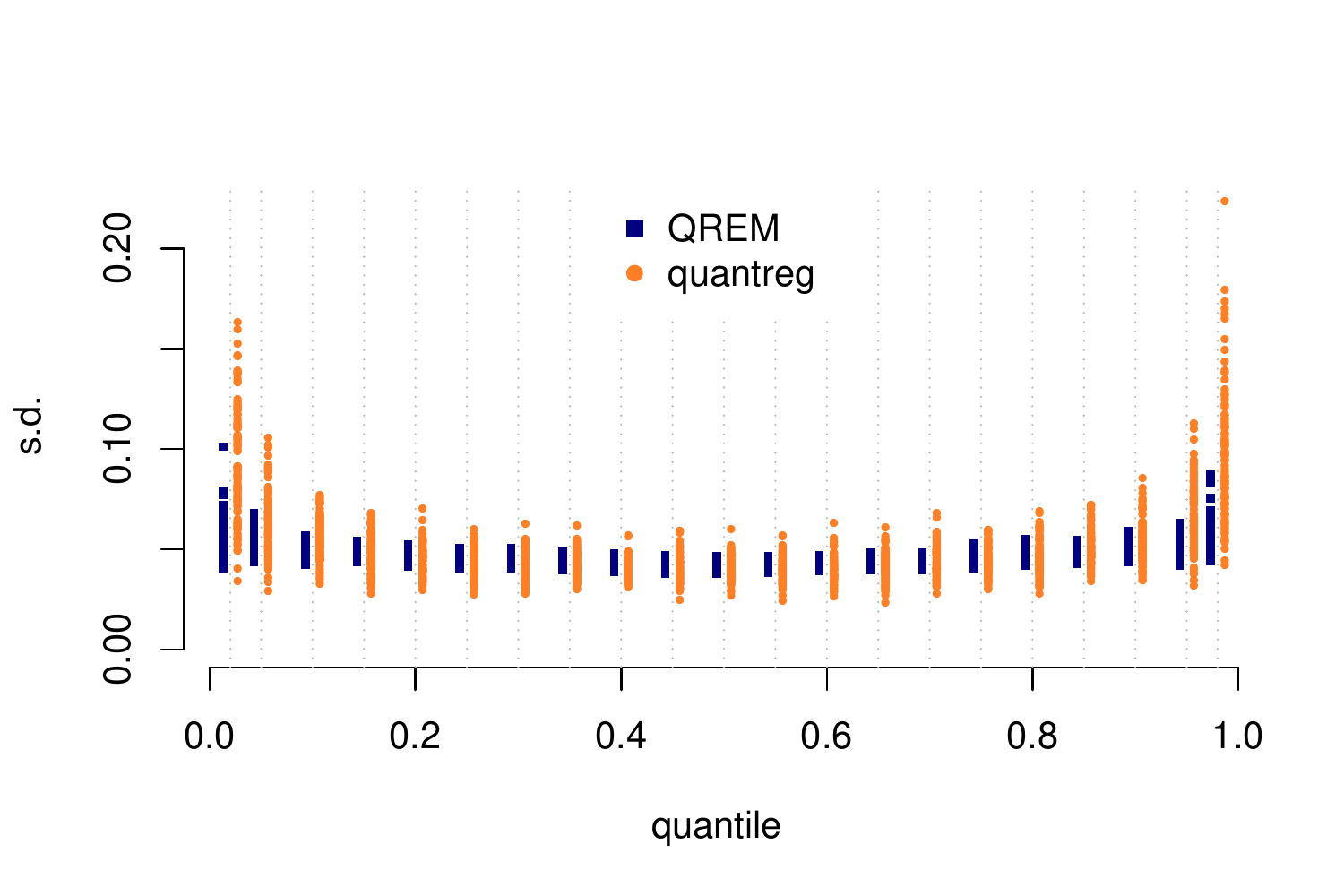}
\end{center}
\caption{QR simulation \#15: $\hat\sigma_{\hat\beta_1}$ for $q\in\{0.02,0.05,0.1,\ldots,0.95,0.98\}$ for
$y\sim N(5+x,(0.1+0.5x)^2)$.}\label{sdbeta1sim15}
 \end{figure}

\begin{figure}[b!]
  \begin{center}
 \includegraphics[trim={0 0.5cm 2 1cm},clip, scale=0.7]{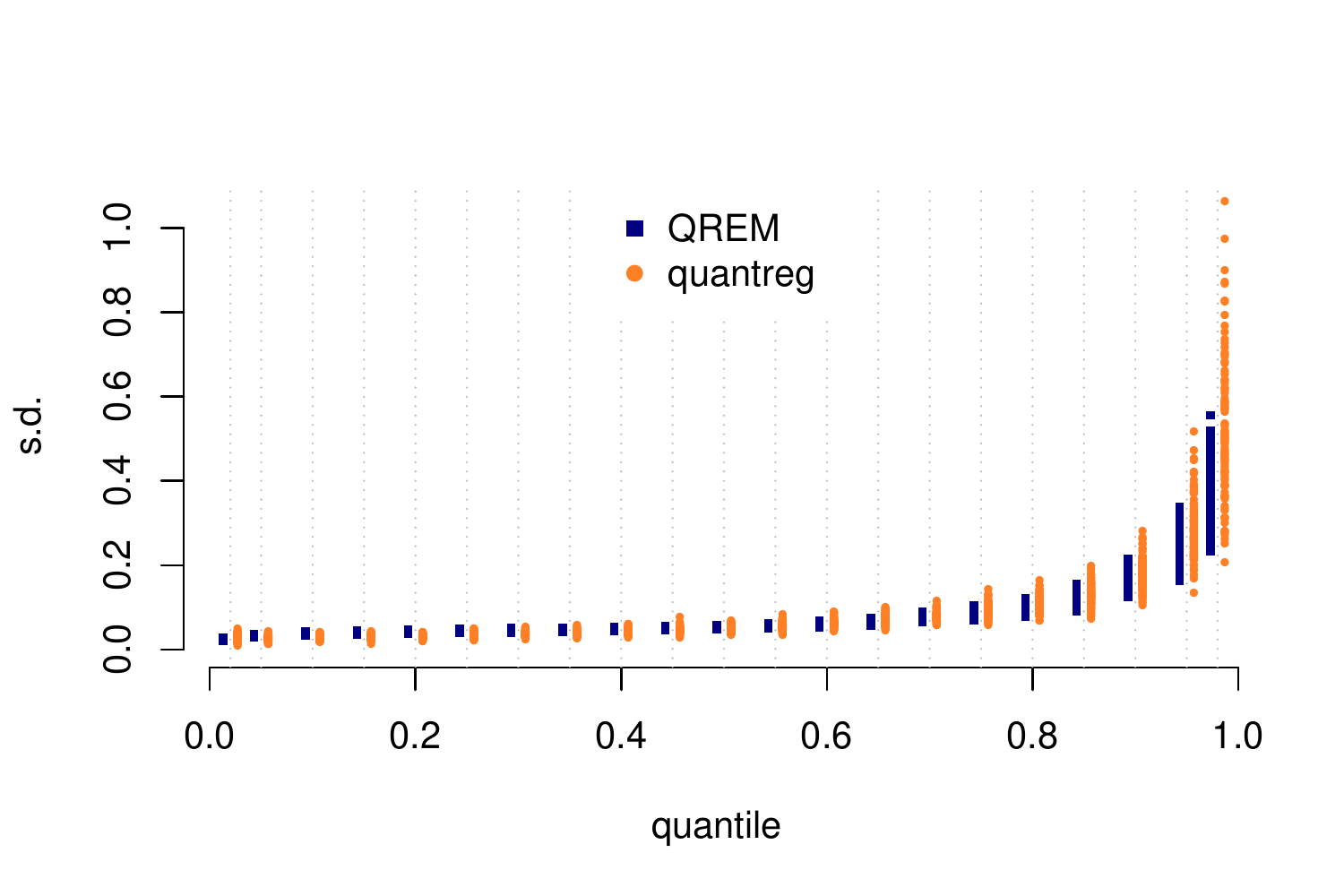}
\end{center}
\caption{QR simulation \#21: $\hat\sigma_{\hat\beta_1}$ for $q\in\{0.02,0.05,0.1,\ldots,0.95,0.98\}$ for
$y=1-3x_1+2x_2+2x_3-x_4-2x_5+\epsilon_i$ and $\epsilon_i\sim LN(0,0.75)$.}\label{sdbeta1sim21}
 \end{figure}

%






\bibliographystyle{chicago}
\bibliography{qrem}

\end{document}